\documentclass{article}
\usepackage{balance}
\usepackage{mathrsfs}
\usepackage{multirow}
\usepackage{multicol}
\usepackage{amsmath}
\usepackage{amsthm}
\usepackage[vlined,ruled,linesnumbered,boxed]{algorithm2e}
\usepackage{algorithmic}
\usepackage{amssymb}
\usepackage{appendix}
\usepackage{color}
\usepackage{amstext}
\usepackage{hyperref}
\usepackage{array}
\usepackage{authblk}
\usepackage{geometry}
\usepackage{graphicx}
\newcommand{\Z}{\mathbb{Z}}
\newcommand{\Q}{\mathbb{Q}}
\newcommand{\N}{\mathbb{N}}
\newcommand{\R}{\mathbb{R}}

\newcommand{\C}{\mathbb{C}}
\newcommand{\I}{\sqrt{-1}}

\newcommand{\lc}{\mathop{\mathtt{LC}}}

\newcommand{\discrim}{\mathop{\mathtt{discrim}}}
\newtheorem{theorem}{Theorem}[section]
\newtheorem{proposition}[theorem]{Proposition}
\newtheorem{lemma}[theorem]{Lemma}
\newtheorem{corollary}[theorem]{Corollary}
\newtheorem{example}[theorem]{Example}
\newtheorem{remark}[theorem]{Remark}

\newtheorem{definition}[theorem]{Definition}

\title{Computing Bunches of
Semi-Periodic Solutions of Bivariate
Exponential-Trigonometric Polynomial Equations with
Separated Variables}
\author{Tao Zheng\thanks{zhengtao@xidian.edu.cn}, Hao Yuan}
\affil{School of Mathematics and Statistics, Xidian University}
\date{}

\begin{document}
{\fontsize{12}{14}\selectfont
\maketitle
\begin{abstract}
A bivariate exponential-trigonometric polynomial (BETP) equation with separated variables is of the form 
\[g(x, e^x, y, \sin y, \cos y) = 0\] 
with $g$ a polynomial and $x$, $y$ real variables. Solving BETP equations with separated variables is useful in engineering. Besides, the problems of computing complex roots of rational-coefficient mixed-trigonometric polynomials (MTP) and exponential polynomials, which occur frequently in dynamic systems, can both be reduced to solving a system containing two BETP equations with separated variables:
\[
\left\{\begin{array}{rcl}
    g(x, e^x, y, \sin y, \cos y) & =&0 \\
     h(x, e^x, y, \sin y, \cos y)& = & 0
\end{array}\right.
.\]
The main contributions of this paper are as follows: First, the theory of a class of functions called the analytic algebraic exponential polynomials is developed. Based on the terminal properties of these functions, we show that if some non-degenerate conditions hold for the system above, then there are real numbers $N>0$ and $M>0$ such that all solutions of that system in the quarter
$\{(x, y)\in\R^2| x>N, y>M \}$ lie on the curves of finitely many analytic algebraic exponential polynomials which are increasing and tend to infinity. These solutions consist of finitely many bunches of so-called semi-periodic solutions, and each bunch is entirely distributed along a certain curve. Finally, effective algorithms have been implemented to find those curves and to count those bunches of semi-periodic roots. 
\end{abstract}
\section{Introduction}
Solving algebraic and transcendental equations are fundamental problems in the field of scientific computation. This paper focuses on symbolic methods for computing solutions to transcendental equations rather than numerical ones.

Real root isolation is a symbolic method for computing real roots of univariate polynomials. It is fundamental in computational real algebraic geometry and implemented in computer algebra systems like Maple and Mathematica. Real root isolation has also been studied for non-polynomial functions. In \cite{mitchell1990robust}, interval arithmetic was used to isolate roots of analytic functions. More recently, Strzeboński proposed methods for exp-log-arctan functions \cite{strzebonski2008real,strzebonski2012real}, using weak/semi-Fourier sequences (related to concepts in \cite{richardson1991towards}). He also studied tame elementary functions via false derivatives \cite{strzebonski2009real}, with completeness relying on Schanuel's conjecture.

In \cite{achatz2008deciding,mccallum2012deciding}, Achatz, McCallum, and Weispfenning developed a recursive real root isolation algorithm for generalized integral polynomials in trans$(x)$ (i.e., $f(x,\text{trans}(x))$), using pseudo-differentiation, Rolle’s theorem \cite{collins1976polynomial}, and Lindemann’s result \cite{shidlovskii2011transcendental}. Their method applies to $\exp(x),\ln(x),\arctan(x)$, with the $\exp(x)$ case implemented and believed more efficient than \cite{maignan1998solving}. Inspired by \cite{achatz2008deciding} and \cite{strzebonski2011cylindrical}, Xu \emph{et al.} studied a class of quantified exponential polynomial formulas extending polynomial ones \cite{xu2015quantifier}, modifying the isolation algorithm from \cite{achatz2008deciding}. Later, Huang et al. isolated positive roots of poly-powers \cite{huang2018positive}, while Wang and Xu addressed radical expressions \cite{wang2021symbolic}. Moreover, related transcendental decision problems over the reals can be find in \cite{chen2024reduction,lipparini2025satisfiability,ni2024local,lipparini2024satisfiability}.

Among those non-polynomials that have been studied, there is a class of transcendental functions called {\em mixed trigonometric-polynomials} (MTPs). An MTP refers to a function of the form \[f(x,\sin x, \cos x)\] with $f$ a trivariate polynomial with rational coefficients. In \cite{chen2024isolating}, Chen \emph{et al.} developed the first algorithm isolating all the roots of an MTP on the whole real axis and proposed the concept of ``periodic root". This research is restricted to the real roots. However, in many real-word cases we encounter some differential equation whose character equation (can be an MTP or an exponential polynomial) is with rational coefficients and we want to study its complex roots in order to decide the stability of the original differential equation (see \cite{matsumoto2012nonlinear}).

In the following, we introduce how the problem of computing the complex roots of an MTP $f(z,\sin z,\cos z)$ reduces to solving the real solutions of a group of BETP equations: We consider equivalently the complex roots of the function 
\[f\big(z\I,\sin\,(z\I),\cos\,(z\I)\big).\]
 Note that 
\[
\begin{array}{rcl}\vspace{2mm}
     &&f\big(z\I,\sin\,(z\I),\cos\,(z\I)\big)  \\\vspace{2mm}
     &=&f(z\I,\frac{e^z-e^{-z}}{2}\I,\frac{e^z+e^{-z}}{2})\\\vspace{2mm}
     &\in&\mathbb{Q}[\I][z,e^z,e^{-z}].
\end{array}
\]
There is a positive integer $k$ such that the function 
\begin{equation}\label{equ:ine^z}
h(z)=f\big(z\I,\sin\,(z\I),\cos\,(z\I)\big)\cdot (e^{z})^k
\end{equation} is in the ring $\mathbb{Q}[\I][z,e^z]$ 
and it shares the same roots with the function 
\begin{equation}\label{equ:f(iz)}
    f\big(z\I,\sin\,(z\I),\cos\,(z\I)\big).
\end{equation}
Replacing the complex variable $z$ in the function $h(z)$ in formula (\ref{equ:ine^z}) by $a+b\I$ with $a$ and $b$ being real variables, we see that both the real and the imaginary parts of $h(z)$ belong to the function ring $\mathbb{Q}[a,b,e^a,\sin b, \cos b]$. Suppose that $r, s\in\mathbb{Q}[x_1,\dots,x_5]$ are polynomials such that $r(a,b,e^a,\sin b, \cos b)$ and $s(a,b,e^a,\sin b, \cos b)$ equal respectively to the real and the imaginary parts of $h(z)$, then finding a root of $h(z)$ is equivalent to finding a solution $(a,b)$ to the following BETP equations
\begin{equation}\label{equ:ri}\left\{
\begin{array}{rcl}
     r(a,b,e^a,\sin b, \cos b)&=&0  \\
     s(a,b,e^a,\sin b, \cos b)&=&0
\end{array}.
\right.
\end{equation}

The goal of this paper is to describe the real solutions of equations (\ref{equ:ri}) in the quarter 
\[
\{(x,y)\in\R^2\,|\,x>N,y>M\},
\]
with $N$ and $M$ being  sufficiently large positive numbers, via the concept of semi-periodic roots (Definition \ref{定义：半周期根}). The main conclusions are as follows: All the real solutions in that quarter are on the graphs of some finitely many analytic algebraic exponential polynomials (Definition \ref{定义：消失多项式和代数指数多项式}) $\phi^{(1)}<\phi^{(2)}<\cdots<\phi^{(t)}$, which are some functions  defined on the interval $(N,+\infty)$ that tend to $+\infty$. Moreover, the set of the solutions in that quarter can be expressed as a disjoint union of finitely many bunches of semi-periodic roots, each of which is entirely on the graph of a certain function (but there can be more than one bunches on one graph). A bunch of semi-periodic solutions consists of countably many solutions, the $b$-coordinates of which are nearly periodic while the $a$-coordinates are not (Remark \ref{评论：为什么叫做半周期解}). This is why they are called \emph{ a bunch of semi-periodic solutions}. Finally, algorithms have been implemented to compute these functions $\phi^{(i)}$ and the number of those bunches of semi-periodic solutions (Algorithm \ref{算法：半周期根计数}), although the specific methods for calculating the numbers $N$ and $M$ have not yet been addressed in this paper.

The rest of the paper is organized as follows.
The next section first introduces some properties about the resultant and then establishes some basic properties of the exponential polynomials. In Section \ref{节：解析代数指数多项式最终性质}, important properties of the algebraic exponential polynomials and their derivatives are proved. In Section \ref{节：根函数的精细最终性质} we introduce delicate properties of the so-called \emph{root-functions}, which are at the core of how we locate the semi-periodic solutions in Section \ref{sec:semiper}. Finally, Section \ref{节：算法和实验} shows the algorithms counting the bunches of semi-periodic solutions and the corresponding experiments.

\section{Preliminaries}
\subsection{Basic Properties on Resultants}
This subsection contains some properties on resultants that will be used later.
\begin{proposition}\label{prop:bukeyueJieshiFeiling}
    If $g,h\in\mathbb{Q}[x_1,x_2,x_3]$ are co-prime and both contain $x_3$, then $\text{res}_{x_3}(g,h)\neq0.$
\end{proposition}
\begin{proof}
    If $\text{res}_{x_3}(g,h)=0$, then $g$  and $h$ have a common factor whose degree with respect to $x_3$ is positive, contradicting the co-prime assumption.
\end{proof}
\begin{proposition}\label{prop:bukeyuePanbieshiFeiling}
    If $g\in\mathbb{Q}[x_1,x_2,x_3]$ is irreducible with $\deg(g,x_3)>0$, then $\discrim(g,x_3)\neq0.$
\end{proposition}
\begin{proof}
    The only non-constant factor of $g$ is $g$ itself, while $\deg(\partial g/\partial x_3,x_3)<\deg(g,x_3)$ indicating $\partial g/\partial x_3\neq g$. Thus $g$ and $\partial g/\partial x_3$ are co-prime. Their resultant and hence $\discrim(g,x_3)$ are both nonzero.
\end{proof}
\begin{proposition}\label{prop:panbieshiFeilingChongyao}
    For $g\in\mathbb{Q}[x_1,x_2,x_3]$ with $\deg(g,x_3)>0$, then $\discrim(g,x_3)\neq0$ iff
    \begin{equation}\label{equ:x3Wupingfang}
    g=p_1^{n_1}(x_1,x_2)\cdots p_k^{n_k}(x_1,x_2)q_1(x_1,x_2,x_3)q_2(x_1,x_2,x_3)\cdots q_m(x_1,x_2,x_3)
    \end{equation}
    with $k\geq0$, $m\geq1$, $n_i\geq0$, and $p_i$, $q_j$ different irreducible factors.
\end{proposition}
\begin{proof}
    This is to say $\discrim(g,x_3)=0$ iff there is an irreducible factor of $g$ containing $x_3$ with an exponent $\geq2$ in the factorization of $g$.
    
    If such factor does exist, it is a   common factor of $g$ and $\partial g/\partial x_3$, implying that $\discrim(g,x_3)=0$.

    If such factor does not exist, then $g$ has the form in equation (\ref{equ:x3Wupingfang}). One verifies by contradiction that every factor of $g$ containing $x_3$ is not a factor of $\partial g/\partial x_3$. Thus $g$ and $\partial g/\partial x_3$ have no common factor containing $x_3$, implying that $\discrim(g,x_3)\neq0$.
\end{proof}
\begin{lemma}\label{lemma:mouyinziZhihanx_1}
    Suppose $p\in\mathbb{Q}[x_1,x_2,x_3]$ is irreducible with $\deg(p,x_3)>0$. If $p(x_1,x_2,x_1x_3)=gh$ for some $g,h\in\mathbb{Q}[x_1,x_2,x_3]$, then one of $g$ and $h$ equals $cx_1^n$ with $c\in\mathbb{Q}^*$ and $n\geq 0.$ 
\end{lemma}
\begin{proof}
    Taking $x_3=t/x_1$ in the equation $p(x_1,x_2,x_1x_3)=gh$, we obtain an equation in the field $\mathbb{Q}(x_1,x_2,t)$:
    \[p(x_1,x_2,t)=g(x_1,x_2,t/x_1)h(x_1,x_2,t/x_1).
    \]
    Suppose the minimal exponent of $x_1$ in those terms of $g(x_1,x_2,t/x_1)$ is $\ell_g\in\mathbb{Z}$, we may define $m=\max\{-\ell_g,0\}$. A similar number $\ell_h$ for $h(x_1,x_2,t/x_1)$ can also be defined and we set $k=\max\{-\ell_h,0\}$. Then we have the following equation in $\mathbb{Q}[x_1,x_2,t]$:
    \begin{equation}\label{equ:fenmuFenjie}
x_1^{m+k}p(x_1,x_2,t)=(x_1^mg(x_1,x_2,t/x_1))(x_1^kh(x_1,x_2,t/x_1)),
    \end{equation}
with $x_1^mg(x_1,x_2,t/x_1)\triangleq G$ and $x_1^kh(x_1,x_2,t/x_1)\triangleq H$ polynomials in $\mathbb{Q}[x_1,x_2,t]$. Since $p(x_1,x_2,t)$ is irreducible, we have $p(x_1,x_2,t)|G$ or $p(x_1,x_2,t)|H.$ Suppose $p(x_1,x_2,t)|G$, from equation (\ref{equ:fenmuFenjie}) we have 
    \begin{equation}
x_1^{m+k}=\frac{G}{p(x_1,x_2,t)}H,
    \end{equation}
which means that  $H|x_1^{m+k}$. Therefore, \[h(x_1,x_2,t/x_1)=H/x_1^k=cx_1^n\] with $c\in\mathbb{Q}^*$ and $m\geq n\in\mathbb{Z}$. But this indicates that the polynomial $h(x_1,x_2,x_3)$ does not contain $x_2$ and $x_3$. Thus $h(x_1,x_2,x_3)=cx_1^n$ with $n\geq0$.
\end{proof}
This lemma indicates the proposition below:
\begin{proposition}\label{prop:qx1x3Fenjie}
    Suppose $p\in\mathbb{Q}[x_1,x_2,x_3]$ is irreducible with $\deg(p,x_3)>0$, then \[p(x_1,x_2,x_1x_3)=x_1^nf(x_1,x_2,x_3)\]
    with $n\geq0$ and $f$ irreducible.
\end{proposition}
\begin{proof}
    Let $n\geq0$ be the nonnegative integer such that $x_1^n|p(x_1,x_2,x_1x_3)$ but $x_1^{n+1}\nmid p(x_1,x_2,x_1x_3)$. Set $f=p(x_1,x_2,x_1x_3)/x_1^n$, then $x_1\nmid f$. We need to show that $f$ is irreducible.

    Assume on the contrary that $f=f_1f_2$. Then, since $x_1\nmid f$, neither $f_1$ nor $f_2$ is of the form $cx_1^m$, where $c\in\mathbb{Q}^*$ and $m\geq0$. Therefore, we have 
    \[p(x_1,x_2,x_1x_3)=(x_1^nf_1)\cdot f_2\]
    with two factors $(x_1^nf_1)$ and $f_2$. Neither of them is of the form $cx_1^m$ with $c\in\mathbb{Q}^*$ and $m\geq0$. This contradicts Lemma \ref{lemma:mouyinziZhihanx_1}.
\end{proof}
The following proposition is useful:
\begin{proposition}\label{prop:x1x3Baochix3wupingfang}
Suppose $g(x_1,x_2,x_3)\in\mathbb{Q}[x_1,x_2,x_3]$ satisfies $\deg(g,x_3)>0$  and $\discrim(g,x_3)\neq0$, then the polynomial $h(x_1,x_2,x_3)=g(x_1,x_2,x_1x_3)$ also satisfies these conditions.
\end{proposition}
\begin{proof}
    It is clear that $\deg(h,x_3)>0$. To prove that $\discrim(h,x_3)\neq0$ we need to use the factorization 
    \[    g(x_1,x_2,x_3)=p_1^{n_1}(x_1,x_2)\cdots p_k^{n_k}(x_1,x_2)q_1(x_1,x_2,x_3)\cdots q_m(x_1,x_2,x_3)\]
    in Proposition \ref{prop:panbieshiFeilingChongyao}.

    It suffices to show that there is no irreducible polynomial $f(x_1,x_2,x_3)$ with $\deg(f,x_3)>0$ such that $f^2|q_i(x_1,x_2,x_1x_3)$ for some $q_i$, or such that $f|q_i(x_1,x_2,x_1x_3)$ and $f|q_j(x_1,x_2,x_1x_3)$ for some pair $q_i,\,q_j$ with $i\neq j$.

    The factorization given in Proposition \ref{prop:qx1x3Fenjie} already makes  the condition \[f^2|q_i(x_1,x_2,x_1x_3)\] impossible.
When $f|q_i(x_1,x_2,x_1x_3)$ and $f|q_j(x_1,x_2,x_1x_3)$, by Proposition \ref{prop:qx1x3Fenjie} we have 
    \[
    \left\{
    \begin{array}{rcl}
        q_i(x_1,x_2,x_1x_3) &=&x_1^{\ell_i}\cdot(c_if)\\
         q_j(x_1,x_2,x_1x_3)& = &x_1^{\ell_j}\cdot(c_jf)
    \end{array}
    \right.
    \]
    with $\ell_i,\ell_j\geq0$ and $c_i,c_j\in\mathbb{Q}^*$. Taking $x_3=t/x_1$, we have the following equations in the field $\mathbb{Q}(x_1,x_2,t)$:
    \begin{equation}\label{equ:qiqjFenjie}
    \left\{
    \begin{array}{rcl}
        q_i(x_1,x_2,t) &=&c_ix_1^{\ell_i}f(x_1,x_2,t/x_1)\\
         q_j(x_1,x_2,t)& = &c_jx_1^{\ell_j}f(x_1,x_2,t/x_1)
    \end{array}.
    \right.
    \end{equation}
Suppose the minimal exponent of $x_1$ in the terms of $f(x_1,x_2,t/x_1)$ is $\ell\in\mathbb{Z}$, then we have $\ell_i\geq\max\{-\ell,0\}$ since $c_ix_1^{\ell_i}f(x_1,x_2,t/x_1)=q_i(x_1,x_2,t)$ is a polynomial. If $\ell_i>\max\{-\ell,0\}$, then $c_ix_1^{\ell_i}f(x_1,x_2,t/x_1)$ is divided by $x_1$. However, $x_1\nmid q_i(x_1,x_2,t)$ since $q_i$ is irreducible and $\deg(q_i(x_1,x_2,t),t)>0$. Therefore, we have $\ell_i=\max\{-\ell,0\}$. Similarly, $\ell_j=\max\{-\ell,0\}$ as well. Then equations (\ref{equ:qiqjFenjie}) contradicting the condition that  $q_i(x_1,x_2,t)$ and $q_j(x_1,x_2,t)$ are co-prime.
\end{proof}
\subsection{Terminal Properties of Exponential Polynomials}
The function ring of exponential polynomials \[\mathbb{Q}[a,e^a]=\big\{p(a,e^a)\,|\,p\in\mathbb{Q}[x_1,x_2]\big\}\]
with each element a function in the variable $a$, is important in this paper. We prove some propositions about these functions:
\begin{proposition}\label{prop:non0}
    For any non-zero polynomial $c\in\mathbb{Q}[x_1,x_2]$, the function $c(a,e^a)$ tends to $\infty$, $-\infty$ or a non-zero constant as $a\rightarrow \infty$.
\end{proposition}
\begin{proof}
    We define a lexicographic order between the monomials of $c(a,e^a)$ as follows
    \[a^\ell(e^a)^m\succ a^{\ell'}(e^a)^{m'}\text{\;\; iff \;\;\;}m>m'\lor (m=m'\land \ell>\ell'),\]
for any $\ell,m,\ell',m'\in\mathbb{Z}_{\geq0}\in$. Then we have \[c(a,e^a)=\lambda a^\ell(e^a)^m+\cdots,\] where $\lambda\in\mathbb{Q}$ and $a^\ell(e^a)^m\; (\text{with }\ell,m\in\mathbb{Z}_{\geq0})$ is the maximal monomial. One observes that the limitation of $c(a,e^a)$ only depends on the maximal term $\lambda a^\ell(e^a)^m$. That is, 
\begin{equation}
\lim_{a\rightarrow\infty}c(a,e^a)=\left\{
    \begin{array}{cc}
       \text{sgn}(\lambda)\cdot\infty  &\text{, if }\ell+m>0\\
       \lambda&\text{, if }\ell=m=0
    \end{array}
    \right.,
\end{equation}
from which the conclusion follows.
\end{proof}
The corollaries below are then straightforward:
\begin{corollary}\label{coro:nonzeroforlarge}
    For any non-zero bivariate polynomial $c$ with coefficient in $\mathbb{Q}$, the value of  $c(a,e^a)$ is non-zero for sufficiently large positive $a$. Moreover, either $c(a,e^a)>0$ for sufficiently large positive $a$ or $c(a,e^a)<0$ for sufficiently large positive $a$.
\end{corollary}

\begin{corollary}
    For any bivariate polynomial $c$ with coefficient in $\mathbb{Q}$, $c$ is a zero polynomial iff the function $c(a,e^a)$ is identically zero.
\end{corollary}
We call Proposition \ref{prop:non0} and Corollary \ref{coro:nonzeroforlarge} ``the terminal properties" since they describe the behavior of the functions when $a\rightarrow \infty$.

\section{Terminal Properties of Algebraic Exponential Polynomials}\label{节：解析代数指数多项式最终性质}
Let $r,\,s\in\Q[x_1,x_2,\ldots,x_5]$ be polynomials as in equations (\ref{equ:ri}). Set $\mathfrak{o}=x_4^2+x_5^2-1$, then
\begin{equation}
\left\{\begin{array}{cc}
    p= & \text{res}(\text{res}(r,\mathfrak{o},x_4),\text{res}(s,\mathfrak{o},x_4),x_5)\in\mathbb{Q}[x_1,x_2,x_3] \\
    q= & \text{res}(\text{res}(r,\mathfrak{o},x_4),\text{res}(s,\mathfrak{o},x_4),x_2)\in\mathbb{Q}[x_1,x_3,x_5]
\end{array}\right..
\end{equation}
 According to the properties of the resultant operation, both $p$ and $q$ belong to the ideal $\langle r,s,\mathfrak{o}\rangle$ generated in the ring $\mathbb{Q}[x_1,x_2,\ldots,x_5]$. Hence equations (\ref{equ:ri}) imply 
 \begin{equation}\label{equ:pq}
\left\{\begin{array}{lcl}
    p(a,b,e^a)&= & 0 \\
    q(a,e^a,\cos b)&= &0
\end{array}\right..
\end{equation}
This means every solution $(a,b)$ to equation (\ref{equ:ri}) is a solution to equation (\ref{equ:pq}).

The equations $p(a,b,e^a)=0$ and $q(a,e^a,\cos b)=0$ indicate, respectively, that $b$ and $\cos b$ can be ``seen" as algebraic functions about $a$ over the ring $\Q[a,e^a]$ in some sense. Therefore, in this section we study functions which are algebraic over the ring $\mathbb{Q}[a,e^a]$ of exponential polynomials. Before giving a formal definition of that, we need the concept of a terminal function:
\begin{definition}
    A function $\theta: S\rightarrow\R$ is called \emph{a terminal function} if $[M,\infty)$ is a subset of its domain $S$ for some $M>0$. The set of terminal functions is denoted by $\mathscr{T}$. The \emph{germ} of a terminal function $\theta$ is defined by 
    \[
        [\theta]=\{\vartheta\in\mathscr{T}\;|\;\exists M>0\;\forall a\geq M,\;\vartheta(a)=\theta(a)\}.
    \]
    
\end{definition}
Then we have the following definitions which will be used from time to time:
\begin{definition}\label{定义：消失多项式和代数指数多项式}
For a terminal function $\theta$, a polynomial $g(x_1,x_2,x_3)$ in $\Q[x_1,x_2,x_3]$ is called an \emph{annihilating polynomial} of $\theta$ if $g$ is  nonzero and $g(a,e^a,\theta(a))\equiv0$ in the domain of $\theta$. A terminal function with an annihilating polynomial is called \emph{an algebraic exponential polynomial}.
\end{definition}

\begin{definition}\label{def:order}
The order $x_2\succ x_1\succ x_3$ between the variables $x_2, x_1, x_3$ defines a pure lexicographic ordering $\mathscr{O}$ between all tri-variate monomials generated by $x_1, x_2$ and $x_3$, which is called \emph{the exp-leading order}. We know that $\mathscr{O}$ is a well ordering. A \emph{leading monomial} of a nonzero polynomial in $\Q[x_1,x_2,x_3]$ is then defined as its maximal monomial with respect to the ordering $\mathscr{O}$. An annihilating polynomial of an algebraic exponential polynomial $\theta$ is called minimal if it has the minimal leading monomial among those annihilating polynomials of $\theta$. 
\end{definition}

The proposition below describes the terminal property of an analytic algebraic exponential polynomial.
\begin{proposition}\label{prop:rootfunclarge}
    Suppose $\theta$ is an analytic algebraic exponential polynomial with $g\in\Q[x_1,x_2,x_3]\backslash\{0\}$ its annihilating polynomial. Then either $\theta(a)\equiv0$ or there is an $M>0$ so that $\theta(a)\neq0$ for every $a>M$.
\end{proposition}
\begin{proof}
\textcolor{blue}{Case 1:} when $x_3$ does not divide $g$.

Assume that $a_0\in\R$ is a root of $\theta$. Then $0=g(a_0,e^{a_0},\theta(a_0))=g(a_0,e^{a_0},0)$. Sine $x_3\not|g$, $g(x_1,x_2,0)$ is not zero polynomial. Corollary \ref{coro:nonzeroforlarge} indicates that there is an $M>0$ such that $g(a,e^{a},0)\neq0$ for every $a\geq M$. Therefore $a_0<M$. But $a_0$ is an arbitrary picked root of $\theta$, so $\theta(a)\neq0$ for every $a\geq M$.

\textcolor{blue}{Case 2:} when $x_3$ divides $g$.

We rewrite $g=x_3^\ell\cdot\hat{g}$ with $\ell\in\Z_{>0}$ and $x_3$ not dividing $\hat{g}$. Since $\theta$ is terminal, there is an $N>0$ so that $\theta(a)^\ell\hat{g}(a,e^a,\theta(a))=0$ for every $a\geq N$.
Define $A=\{a\geq N\;|\;\theta(a)=0\}$ and $B=\{a\geq N\;|\;\hat{g}(a,e^a,\theta(a))=0\}$, then $A\cup B=[N,\infty)$. Therefore, any $a_0> N$ is an accumulation point of $A$ or $B$. But both $\theta(a)$ and $\hat{g}(a,e^a,\theta(a))$ are analytic, the existence of an accumulation point of $A$ or $B$ implies $\theta(a)\equiv0$ or $\hat{g}(a,e^a,\theta(a))\equiv0$, respectively. The latter case where $\hat{g}(a,e^a,\theta(a))\equiv0$ goes back to \textcolor{blue}{Case 1} of the proof since $x_3$ does not divide $\hat{g}$.
\end{proof}
We then have the following corollary:
\begin{corollary}\label{coro:posIfx3notDiv}
    Suppose $\theta$ is an analytic algebraic exponential polynomial with $g\in\Q[x_1,x_2,x_3]\neq0$ its annihilating polynomial and $x_3$ does not divide $g$. Then, either there is an $M>0$ so that $\theta(a)>0$ for every $a>M$, or there is an $M>0$ so that $\theta(a)<0$ for every $a>M$.
\end{corollary}
\begin{definition}\label{def:rootfunctionfield}
The \emph{field of germs} of analytic algebraic exponential polynomials is defined by
\begin{equation}\label{equ:germField}
\mathscr{R}=\{[\theta]\;|\;\theta\text{ is an analytic algebraic exponential polynomial}\}.
\end{equation}
\end{definition}

The following proposition claims that the ``field" of germs of analytic algebraic exponential polynomials $\mathscr{R}$ is indeed a field.
\begin{proposition}\label{prop:isfield}
    The set $\mathscr{R}$ of germs  of analytic algebraic exponential polynomials in Definition \ref{def:rootfunctionfield} is a field.
\end{proposition}
\begin{proof}

Suppose $\theta$ and $\vartheta$ are analytic algebraic exponential polynomials with minimal annihilating polynomial $g$ and $h$. Then both $g$ and $h$ contain $x_3$, otherwise $0\equiv g(a,e^a,\theta(a))\equiv g(a,e^a)$ for a nonzero polynomial $g$, contradicting Corollary \ref{coro:nonzeroforlarge}.

The germs of the constant functions $0$ and $1$ belong to $\mathscr{R}$ and they are the identities of the addition and the multiplication in $\mathscr{R}$ defined below:
\[\left\{\begin{array}{rcl}
     [\theta]\pm[\vartheta]&=&[\theta\pm\vartheta] \\
     
    [\theta]\cdot[\vartheta] & =&[\theta\cdot\vartheta]\\
    
    [\theta]/[\vartheta] & =&[\theta/\vartheta]\text{,\;\;\; if } [\vartheta]\neq[0]
\end{array}\right..
\]
These formulae are well-defined: the point-wise operations $\theta\pm\vartheta, \theta\cdot\vartheta$ and $\theta/\vartheta$ are executed over the intersection of their domains. Moreover, by Proposition \ref{prop:rootfunclarge}, any terminal function $\vartheta(a)$ such that $[\vartheta]\neq[0]$ is nonzero for sufficiently large positive $a$. Hence, the values of $\theta/\vartheta, \theta\pm\vartheta$ and $\theta\cdot\vartheta$ are all well-defined for sufficiently large positive $a$. This is enough for defining new germs.

The key is to find annihilating polynomials for $\theta\pm\vartheta, \theta\cdot\vartheta$ and $\theta/\vartheta$. Noting that $g$ and $h$ are annihilating polynomials for $\theta$ and $\vartheta$, we observes that 
\[\text{res}_w(g(x_1,x_2,x_3\mp w),h(x_1,x_2,w)),\]
\[\text{res}_w(w^{\text{deg}_{x_3}(g)}g(x_1,x_2,x_3/ w),h(x_1,x_2,w)),\]
and
\[\text{res}_w(g(x_1,x_2,x_3\cdot w),h(x_1,x_2,w))\]
are annihilating polynomials of the functions $\theta\pm\vartheta, \theta\cdot\vartheta$ and $\theta/\vartheta$:

It is well-known that they annihilate these functions, the key is to prove that they are nonzero.

Let $\diamond\in\{+,-,\cdot\}$ and $\text{res}_\diamond(x_1,x_2,x_3)=\text{res}_w(g(x_1,x_2,x_3\diamond w),h(x_1,x_2,w))$. Note that $\lc_{x_3}(g)\neq0$ and $\lc_{x_3}(h)\neq0$, if \textcolor{blue}{$h(x_1,x_2,0)\neq0$}, we can chose $x_1=c_1$, $x_2=c_2\in\C$ so that $\lc_{x_3}(g)(c_1,c_2)\neq0$, $\lc_{x_3}(h)(c_1,c_2)\neq0$ and $h(c_1,c_2,0)\neq0$. Set $R_g(c_1,c_2)=\{c\in\C\;|\;g(c_1,c_2,c)=0\}$ and $R_h(c_1,c_2)=\{c\in\C\;|\;h(c_1,c_2,c)=0\}$, then
\begin{equation}\label{equ:resroot}
\text{res}_\diamond(c_1,c_2,x_3)=s\prod_{c\in R_h(c_1,c_2)}g(c_1,c_2,x_3\diamond c)
\end{equation}
with $s=\pm(\lc_{x_3}(h)(c_1,c_2))^{\deg_{x_3}(g)}\neq0$. Since $h(c_1,c_2,0)\neq0$, $R_h(c_1,c_2)$ does not contain $0$. Define
\[
\diamond^{-1}=\left\{\begin{array}{rcl}
     -,& \text{ if }\diamond =+ \\
     +,& \text{ if }\diamond =- \\
     /,& \text{ if }\diamond =\cdot
\end{array}\right.
\]
Then we can chose the value of $x_3=c_3\in\C$ such that $c_3\not\in R_g(c_1,c_2)\diamond^{-1}R_h(c_1,c_2)$. The assignment $x_3=c_3$ then makes the right side of equation (\ref{equ:resroot}) nonzero, meaning that res$_\diamond(c_1,c_2,c_3)\neq0$. Therefore res$_\diamond$ is a nonzero polynomial.

For the multiplication case, we have

\[\text{res}_w(w^{\text{deg}_{x_3}(g)}g(c_1,c_2,x_3/ w),h(c_1,c_2,w))=s\prod_{c\in R_h(c_1,c_2)}(c^{\text{deg}_{x_3}(g)}g(c_1,c_2,x_3/c))\]
with the same $s\neq0$. We can then chose $x_3=c_3\in\C$ so that $c_3\not\in R_g(c_1,c_2)\cdot R_h(c_1,c_2)$, ensuring that $x_3=c_3$ makes the right side of the equation right above nonzero. Therefore, 
\[\text{res}_w(w^{\text{deg}_{c_3}(g)}g(c_1,c_2,c_3/ w),h(c_1,c_2,w))\neq0\] and
\[\text{res}_w(w^{\text{deg}_{x_3}(g)}g(x_1,x_2,x_3/ w),h(x_1,x_2,w))\] is nonzero polynomial.

In a nutshell, we proved that those annihilating polynomials are indeed nonzero when $h(x_1,x_2,0)\neq0$. In the case when \textcolor{blue}{$h(x_1,x_2,0)=0$}, we have $x_3|h$. That is, $h=x_3\hat{h}$ for nonzero polynomial $\hat{h}$, whose leading monomial (w.r.t. $\mathscr{O}$) is inferior to the leading monomial of $h$. We then have $0\equiv h(a,e^a,\vartheta(a))=\vartheta(a)\hat{h}(a,e^a,\vartheta(a))$. By using the same trick in the proof of Proposition \ref{prop:rootfunclarge}, we claim that $\vartheta(a)\equiv0$ or $\hat{h}(a,e^a,\vartheta(a))\equiv0$. But the latter contradicts the minimality of the leading monomial of $h$, thus $\vartheta(a)\equiv0$. Then the annihilating polynomials of the functions $\theta\pm\vartheta$ and $\theta\cdot\vartheta$ (we don't take care of $\theta/\vartheta$) can be just $g$ and $x_3$, which are both nonzero.

In a word, we proved that $[\theta\pm\vartheta], [\theta\cdot\vartheta]$ and $[\theta/\vartheta]$ are in $\mathscr{R}$. 
\end{proof}
\subsection{Algebraic Exponential Polynomials As Root-Functions}
The following theorem shows that two algebraic exponential polynomials can share one annihilating polynomial. In other words, a nonzero polynomial $g\in\mathbb{Q}[x_1,x_2,x_3]$ can annihilate many terminal functions.
\begin{theorem}\label{theorem:rootfunc}
    For any nonzero polynomial $g(x_1,x_2,x_3)$ with coefficient in $\mathbb{Q}$, if $\Delta(x_1,x_2)=\discrim(g,x_3)$ is nonzero polynomial, then there is an $M>0$ and there are some $k\geq0$ analytic functions $\theta_1<\cdots <\theta_k$ from $[M,\infty)$ into $\R$ such that 
\begin{equation}\label{equ:cuHeTuxiang}
    \{(a,y)\in\R^2\;|\;g(a,e^a,y)=0,a\geq M\}=\bigcup_{i=1}^k\Gamma(\theta_i),
\end{equation}
where $\Gamma(\theta_i)=\{(a,\theta_i(a))\;|\;a\geq M\}$ is the graph of $\theta_i$.
\end{theorem}
\begin{proof}
Suppose $\deg(g,x_3)=n$ and $c_n(x_1,x_2)$ is the coefficient of $x_3^n$ in $g$, which is nonzero since $g$ is. Then $c_n(a,e^a)$ is non-zero for sufficiently large positive $a$ by Corollary \ref{coro:nonzeroforlarge}. On the other hand, $\Delta(x_1,x_2)\neq0$ implies $\Delta(a,e^a)>0$ for sufficiently large positive $a$ or $\Delta(a,e^a)<0$ for sufficiently large positive $a$, again by Corollary \ref{coro:nonzeroforlarge}. The rest of the proof is essentially the same to the one of Theorem 2 of  \cite{mccallum1998improved} (see also \cite{chen2024isolating} Proposition 3.2). The key is the fact that $g(a,e^a,y)$ is a polynomial with respect to $y$, although it is not with respect to $a$.
\end{proof}

\begin{definition}\label{def:rootFunctionStart}
The analytic functions $\theta_i$'s are called the \emph{root-functions} of the equation $g(a,e^a,y)=0$, or simply the \emph{root-functions} of the tri-variate polynomial $g$, if there is no ambiguity. The number $M$ is called the \emph{start-point} of those functions.
\end{definition}
\begin{remark}
Clearly, each root-function is an analytic algebraic exponential polynomial annihilated by the nonzero polynomial $g$. Therefore, the terminal properties in Proposition \ref{prop:rootfunclarge} and Corollary \ref{coro:posIfx3notDiv} hold for root-functions.
\end{remark}
The proposition below claims that a continuous terminal function annihilated by a polynomial is exactly one of its root-functions.
\begin{proposition}\label{prop:lianxuTuigenhanshu}
Set $g\in\mathbb{Q}[x_1,x_2,x_3]$ with $\Delta(x_1,x_2)=\discrim(g,x_3)$ nonzero and $\theta_1<\cdots<\theta_k$ its root-functions. If a continuous terminal function $\vartheta$ is annihilated by $g$, then $[\vartheta]=[\theta_i]$ for some $1\leq i\leq k$.
\end{proposition}
\begin{proof}
According to equation (\ref{equ:cuHeTuxiang}), the existence of $\vartheta$ implies that $k\geq1$. 
    
    Suppose that the start-point of these root-functions is $M$. Assume that there are two points on the graph of $\vartheta$ such that they come from the graphs of different root-functions. That is, there are $a_2, a_1\geq M$ both in the domain of $\vartheta$ such that $\vartheta(a_1)=\theta_i(a_1)$ and $\vartheta(a_2)=\theta_j(a_2)$ for some $1\leq i< j\leq k$. Define $\gamma=\vartheta-(\theta_i+\theta_{i+1})/2$, then $\gamma(a_1)=\theta_i(a_1)-(\theta_i(a_1)+\theta_{i+1}(a_1))/2<0$ and \[\gamma(a_2)=\theta_j(a_2)-(\theta_i(a_2)+\theta_{i+1}(a_2))/2\geq\theta_{i+1}(a_2)-(\theta_i(a_2)+\theta_{i+1}(a_2))/2>0.\] Since $\gamma$ is continuous, there is $a_3$ between $a_1$ and $a_2$ such that $\gamma(a_3)=0$, \emph{i.e.}, $\vartheta(a_3)=(\theta_i(a_3)+\theta_{i+1}(a_3))/2$. The point $(a_3,(\theta_i(a_3)+\theta_{i+1}(a_3))/2)$ on the graph of $\gamma$ belongs to the left-hand side of equation (\ref{equ:cuHeTuxiang}) since $\gamma$ is annihilated by $g$, but it does not belong to the right-hand side of that equation because $\theta_i<(\theta_i+\theta_{i+1})/2<\theta_{i+1}$, which is a contradiction.

    Therefore, all points on the graph of $\vartheta$ with $a\geq M$ come from the graph of the same root-function, which means $[\vartheta]=[\theta_i]$ for some $1\leq i\leq k$.
\end{proof}
The following lemma is useful for proving some finer conclusions later.
\begin{lemma}\label{lemma:huzhiTuiGonggongjieYoujie}
 For any two nonzero co-prime polynomials $p$, $q\in\mathbb{Q}[x_1,x_2,x_3]$ and any terminal function $\vartheta$, the solution set of the following equation 
    \[
    \left\{
    \begin{array}{lcr}
         p(a,e^a,\vartheta(a))&=&0 \\
         q(a,e^a,\vartheta(a))&=&0
    \end{array}
    \right.
    \]
    is bounded from above.
\end{lemma}
\begin{proof}
If one of the polynomials $p$ and $q$, say $p$, does not contain $x_3$, then the solution of the equation $p(a,e^a)=0$ is already bounded from above according to Proposition \ref{coro:nonzeroforlarge}.

    Otherwise we have $\deg(p,x_3)>0$ and $\deg(q,x_3)>0$. Assume on the contrary  that there is a sequence $\{a_n\}$ tending to $+\infty$ which is a subset of the solution set of the above equation. Then the resultant 
    \[\text{res}_{x_3}(p(a_n,e^{a_n},x_3),q(a_n,e^{a_n},x_3))=0\] for every $n$ since $p(a_n,e^{a_n},x_3)$ and $q(a_n,e^{a_n},x_3)$ share a common root $x_3=a_n$. This means that the exponential polynomial 
        \[\text{res}_{x_3}(p(a,e^{a},x_3),q(a,e^{a},x_3))\]
        has a sequence of roots $\{a_n\}$ that tends to $+\infty$. By Proposition \ref{coro:nonzeroforlarge}, we claim that 
                \[\text{res}_{x_3}(p(x_1,x_2,x_3),q(x_1,x_2,x_3))=0.\]
    But this indicates that $p$ and $q$ have a common factor containing $x_3$, contradicting the assumption that $p$ and $q$ are co-prime. 
\end{proof}
\begin{proposition}\label{prop:jiexiDaishuzhishuduoxiangshiShiMouyinziGenhanshu}
    Suppose $\vartheta$ is an analytic algebraic exponential polynomials with annihilating polynomial $g\in\mathbb{Q}[x_1,x_2,x_3]$, then $\vartheta$ have the same germ with a root-function of some irreducible factor of $g$ that contains $x_3$.
\end{proposition}
\begin{proof}
    Suppose $p_1,\ldots,p_m$ are all the different irreducible factors of $g$. If we show that the function $\vartheta(a)$, restricted to some interval $(M,+\infty]$, is annihilated by some $p_i$ with $\deg(p_i,x_3)>0$, then Proposition \ref{prop:bukeyuePanbieshiFeiling} guarantees that $\discrim(p_i,x_3)\neq0$ and Proposition \ref{prop:lianxuTuigenhanshu} helps draw the conclusion.
    
    Note that any two different factors $p_i$ and $p_j$ are co-prime. Applying Lemma \ref{lemma:huzhiTuiGonggongjieYoujie} to $p_i$ and $p_j$ we claim that the solution set of the following equation 
    \[
    \left\{
    \begin{array}{lcr}
         p_i(a,e^a,\vartheta(a))&=&0 \\
         p_j(a,e^a,\vartheta(a))&=&0
    \end{array}
    \right.
    \]
    is bounded from above.

    Therefore, there is $M>0$ such that for every $a>M$, there is one and only one factor $p_{\tau(a)}$ among $p_1,\ldots,p_m$ such that $p_{\tau(a)}(a,e^{a},\vartheta(a))=0$. Define 
    \[
    \mathscr{Z}_i=\{a> M\,|\,\tau(a)=i\}
    \]
    for $1\leq i\leq m.$ Then $(M,+\infty)=\cup_{i=1}^m\mathscr{Z}_i$ and this is a pairwise disjoint union. Since $m$ is finite, any sequence in $(M,+\infty)$ tends to some $a\geq M$ has a subsequence contained in some set $\mathscr{Z}_i$ which also tends to $a\geq M$. This means the subset  $\mathscr{Z}_i$ of the root set of the function $p_i(a,e^{a},\vartheta(a))$ has a limit point. Since $\vartheta(a)$ is analytic over $(M,+\infty)$, so is $p_i(a,e^{a},\vartheta(a))$. Thus, $p_i(a,e^{a},\vartheta(a))\equiv0$ over $(M,+\infty)$.

The last thing we need to prove is that $\deg(p_i,x_3)>0$:  If $\deg(p_i,x_3)=0$, then $p_i\in\mathbb{Q}[x_1,x_2]$ and $p_i(a,e^a)=p_i(a,e^a,\vartheta(a))\equiv0$ over $(M,+\infty)$ with $p_i$ nonzero, contradicting Corollary \ref{coro:nonzeroforlarge}.
\end{proof}
The following theorem gives a decomposition for the set of root-functions of a polynomial according to its irreducible factors:
\begin{theorem}\label{thm:genhanshuFenjie}
    Suppose $g\in\mathbb{Q}[x_1,x_2,x_3]$ satisfies $\deg(g,x_3)>0$ and  $\discrim(g,x_3)\neq0$. By Proposition \ref{prop:panbieshiFeilingChongyao}, we have     \[
    g=p_1^{n_1}(x_1,x_2)\cdots p_k^{n_k}(x_1,x_2)q_1(x_1,x_2,x_3)q_2(x_1,x_2,x_3)\cdots q_m(x_1,x_2,x_3)
    \]
     with $k\geq0$, $m\geq1$, $n_i\geq0$ and $p_i$, $q_j$ different irreducible factors satisfying $\deg(q_j,x_3)>0$ for $1\leq j\leq m$. Then the set of the germs of the root-functions of $g$ is equals to the union of the sets of the germs of the root-functions of the factors $q_j$, $1\leq j\leq m$. Moreover, this union is pairwise disjoint.
\end{theorem}
\begin{proof}
    Each root-function of $g$ is an analytic algebraic exponential polynomial annihilated by $g$. Applying Proposition \ref{prop:jiexiDaishuzhishuduoxiangshiShiMouyinziGenhanshu} we claim that each root-function of $g$ has the same germ with one of the root-functions of the factors of $g$ that contains $x_3$.

    For $1\leq j\leq m$, each root-function of the factor $q_j$ is clearly continuous and annihilated by $g$. Applying Proposition \ref{prop:lianxuTuigenhanshu}, we claim that this root-function of the factor $q_j$ has the same germ with one of the root-functions of $g$.

    Lastly, we need to show that for two different factors $q_i$ and $q_j$ with $1\leq i <j \leq m$, if $\theta$ is a root-function of $q_i$ and $\vartheta$ is a root-function of $q_j$, then $[\theta]\neq[\vartheta]$. This follows from Lemma \ref{lemma:huzhiTuiGonggongjieYoujie} since $q_i$ and $q_j$ are co-prime.
\end{proof}

\subsection{Terminal Properties of the Derivative of an Algebraic Exponential Polynomial}
We show first that the derivative of an analytic algebraic exponential polynomial is also an analytic algebraic exponential polynomial.
\begin{proposition}\label{prop:derivativeClose}
    The field of germs of analytic algebraic exponential polynomials is closed under differentiation. That is, $[\theta]\in\mathscr{R}\Rightarrow[\theta']\in\mathscr{R}$ with $\mathscr{R}$ in equation (\ref{equ:germField}).
\end{proposition}
\begin{proof}
Suppose $\theta$ is an analytic algebraic exponential polynomial with $g(x_1,x_2,x_3)$ an annihilating polynomial having minimal degree in $x_3$. If ${\partial g}/{\partial x_3}=0$, then $g\in\mathbb{Q}[x_1,x_2]$ and $g(a,e^a)=g(a,e^a,\theta(a))\equiv0$ with $g$ nonzero, contradicting Corollary \ref{coro:nonzeroforlarge}. Then $g(a,e^a,\theta(a))\equiv0$ implies that
\[\frac{\partial g}{\partial x_1}(a,e^a,\theta(a))+e^a\frac{\partial g}{\partial x_2}(a,e^a,\theta(a))+\theta'(a)\frac{\partial g}{\partial x_3}(a,e^a,\theta(a))\equiv0,\]
with ${\partial g}/{\partial x_3}\neq0$ and $\deg(\frac{\partial g}{\partial x_3},x_3)<\deg(g,x_3)$. Since $g$ is an annihilating polynomial with the minimal degree in $x_3$, the function $\frac{\partial g}{\partial x_3}(a,e^a,\theta(a))$ is not identically zero. By Proposition \ref{prop:isfield}, that function is an analytic algebraic exponential polynomial. By Proposition \ref{prop:rootfunclarge}, there is an $M>0$ so that $\frac{\partial g}{\partial x_3}(a,e^a,\theta(a))\neq0$ for every $a>M$. Thus, we have
\[
\theta'(a)=\frac{-\frac{\partial g}{\partial x_1}(a,e^a,\theta(a))-e^a\frac{\partial g}{\partial x_2}(a,e^a,\theta(a))}{\frac{\partial g}{\partial x_3}(a,e^a,\theta(a))}
\]
for every $a>M$. The right-hand side of the above equation is an analytic algebraic exponential polynomial by Proposition \ref{prop:isfield}, so $[\theta']\in\mathscr{R}$.
\end{proof}
The proposition below describes a terminal property of the derivative of an analytic algebraic exponential polynomial.
\begin{proposition}\label{prop:rootfunclargederi}
    Suppose $\theta$ is an analytic algebraic exponential polynomial. Then either $\theta(a)\equiv c\in\R$ or there is an $M>0$ so that $\frac{\mathrm{d} \theta}{\mathrm{d} a}(a)\neq0$ for every $a>M$.
\end{proposition}
\begin{proof}
By Proposition \ref{prop:derivativeClose}, $\theta'(a)$ is also an analytic algebraic exponential polynomial. Applying Proposition \ref{prop:rootfunclarge} to $\theta'(a)$, we conclude that either $\theta'(a)\equiv0$ or there is $M>0$ so that $\theta'(a)\neq0$ for every $a>M$. In the former case we have $\theta(a)\equiv c\in\mathbb{R}$.
\end{proof}

Proposition \ref{prop:rootfunclargederi} describes the zero-set of the derivative of an analytic algebraic exponential polynomial. The following corollary describe the ``$\lambda$-set" for most of the real numbers $\lambda$.
\begin{corollary}\label{coro:rootfunclargederilambda}
Suppose $\theta$ is an analytic algebraic exponential polynomial with $g\in\Q[x_1,x_2,x_3]$ an annihilating polynomial. Suppose $g(x_1,x_2,x_3+\lambda x_1)\neq0$ for some $\lambda\in\Q$, then either $\theta(a)\equiv\lambda a$ or there is an $M>0$ so that $\frac{\mathrm{d} \theta}{\mathrm{d} a}(a)\neq\lambda$ for every $a>M$.
\end{corollary}
\begin{proof}
    The nonzero polynomial $g(x_1,x_2,x_3+\lambda x_1)$ annihilates the analytic terminal function $\theta(a)-\lambda a$. By Proposition \ref{prop:rootfunclargederi}, either $\theta(a)-\lambda a\equiv0$ or  there is an $M>0$ such that $\frac{\mathrm{d} (\theta(a)-\lambda a)}{\mathrm{d} a}\neq0$ for every $a>M$. These are exactly what we want.
\end{proof}
We conclude from the following proposition that almost all $\lambda\in\Q$ makes $g(x_1,x_2,x_3+\lambda x_1)\neq0$.
\begin{proposition}\label{prop:almostalllambda}
  Let $g\in\Q[x_1,x_2,x_3]$ be a nonzero polynomial, then  $g(x_1,x_2,x_3+\lambda x_1)=0$ for at most finitely many rarional numbers $\lambda\in\Q$.
\end{proposition}
\begin{proof}
    Let $c_{ijk}(\lambda)$ be the coefficient of the polynomial $g(x_1,x_2,x_3+\lambda x_1)$ with respect to the monomial $x_1^jx_2^jx_3^k$, it is a polynomial about $\lambda$. Then $g(x_1,x_2,x_3+\lambda x_1)$ is zero polynomial iff $\lambda$ is a solution to the equations \[\bigwedge_{ijk}c_{ijk}(\lambda)=0.\]
    Since $g\neq0$, $\lambda=0$ is not a solution. Therefore the equations have at most finitely many rational roots.
\end{proof}
\section{Finer Terminal Properties of Root-Functions}\label{节：根函数的精细最终性质}
The following definition is crucial in this section.
\begin{definition}\label{def:OXiashouXishu}
    Set $g\in\mathbb{Q}[x_1,x_2,x_3]$ to be nonzero. Suppose $x_2^{i_0}x_1^{j_0}x_3^{k_0}$ is the maximal monomial of $g$ with respect to the lexicographical order $x_2\succ x_1\succ x_3$, we call the coefficient with respect to the monomial $x_2^{i_0}x_1^{j_0}$ in $g$ \emph{the exp-leading coefficient} of $g$ and denote it by $\lc_{x_2,x_1}(g)$.
\end{definition}
\begin{example}\label{例：指数首系数}
Set \[g(x_1,x_2,x_3)=2x_1x^2_2x_3-7x_1^3x_2+x_1x_2x_3^4+8,\] then the maximal monomial is $x^2_2x_1x_3$ and the exp-leading coefficient $\lc_{x_2,x_1}(g)=2x_3$.

    Set \[
    \begin{array}{rcl}
    \vspace{1mm}
         h(x_1,x_2,x_3)& =&3x_1x_2^2+x_1x^2_2x_3-7x_1x_2+x_1x_2x_3^4+8 \\\vspace{1mm}
         &=&(3+x_3)x_1x_2^2-7x_1x_2+x_1x_2x_3^4+8 
    \end{array}
    ,\] then $x^2_2x_1x_3$ is maximal and the exp-leading coefficient of $h$ is $3+x_3$, that is, the coefficient with respect to $x_2^2x_1$ in $h$.
\end{example}
The exp-leading coefficient of $g$ is a polynomial in $\mathbb{Q}[x_3]$ by definition.

The theorem below shows that the value of a root-function tends to a real constant or to infinity.
\begin{proposition}\label{prop:firstbound}

Let $g\in\Q[x_1,x_2,x_3]$ be a nonzero polynomial with $\lc_{x_2,x_1}(g)$ the exp-leading coefficient and suppose that $i_0$ and $j_0$ are the exponents of $x_2$ and $x_1$, respectively, in the maximal monomial of $g$ with respect to the lexicographical ordering $x_2\succ x_1\succ x_3$. Let $\delta >0$ and let $U \subset \R$ be a bounded set. 

If $|\lc_{x_2,x_1}(g)(y)|>\delta$ for any $y\in U$, then there is an integer $N>0$ such that $a>N$ implies $\forall y\in U,\; g(a,e^a,y)\neq0$.
\end{proposition}
\begin{proof}
Arranging the terms of $g$ decreasingly with respect to the ordering $\mathscr{O}$ and summing up those terms with the same exponents of $x_2$ and $x_1$, we obtain $g=\sum_{ij}c_{ij}(x_3)x_2^ix_1^j$, where $c_{ij}(x_3)\in\mathbb{Q}[x_3]$, and in particular, $c_{i_0j_0}(x_3)=\lc_{x_2,x_1}(g)$. Since $U\subset\R$ is bounded, it is a subset of some bounded closed interval, on which the continuous function $|c_{ij}(y)|$ clearly has an upper bound. Thus there is some $B_{ij}>0$ such that $|c_{ij}(y)|<B_{ij}$ for every $y\in U$. 
Since $i_0$ and $j_0$ are the exponents of the maximal monomial with respect to $\mathscr{O}$, there is an integer $N>0$ such that \[\delta |(e^a)^{i_0}a^{j_0}|>\sum_{i\neq i_0\vee j\neq j_0}B_{ij}|(e^a)^{i}a^{j}|\] for every $a>N$.
These inequalities, together with the assumption that $|\lc_{x_2,x_1}(g)(y)|>\delta\, (\forall y\in U)$ and the triangle inequality, indicate that, for all $a>N$ and all $y \in U$, we have
\[|c_{i_0j_0}(y)(e^a)^{i_0}a^{j_0}|>\delta|(e^a)^{i_0}a^{j_0}|>\sum_{i\neq i_0\vee j\neq j_0}B_{ij}|(e^a)^{i}a^{j}|\geq\left|\sum_{i\neq i_0\vee j\neq j_0}c_{ij}(y)(e^a)^{i}a^{j}\right|.\]
Thus, $a>N$ implies that $g(a,e^a,y)\neq0$ for any $y\in U$.
\end{proof}
By Proposition \ref{prop:firstbound}, every value $y$ such that $g(a, e^a, y)=0$ for some large $a$ ($> N$) necessarily lies in the complementary set of $U$. In
order to locate such pair $(a,y)$, we need to choose a good $U$ in Proposition \ref{prop:firstbound}.
Equivalently, we may choose its complementary set. The following definition
of a potential periodic interval set gives a complementary set we want.
\begin{definition}\label{def:potential-periodic-interval}
For any $g\in\Q[x_1,x_2,x_3]$ containing $x_3$,
let $I$ be a finite list of open intervals $[(a_0,b_0),\ldots,(a_{s+1},b_{s+1})]$, 
where $s\in\Z_{\ge0}$, $a_0=-\infty$, $b_{s+1}=+\infty$, $b_0<0$, $a_{s+1}>0$, $a_j\in \Q~(1\le j\le s+1)$ and $b_j\in \Q~(0\le j\le s)$, such that
\begin{enumerate}
\item  for each $j~(1\le j\le s)$, $a_j< b_j$ and $b_j \le a_{j+1}$;

\item for each real root of $\lc_{x_2,x_1}(g)$, there exists an interval $(a_j,b_j)~(1\le j\le s)$ containing it $($there is no root of $\lc_{x_2,x_1}(g)$ in $(a_0,b_0)$ and $(a_{s+1},b_{s+1}))$; and

\item \label{item:potential-periodic-interval-3} each interval $(a_j,b_j)~(1\le j\le s)$ has exactly one real root of $\lc_{x_2,x_1}(g)$.
\end{enumerate}
We call $I$ an \emph{exponential potential periodic interval set} of the equation \[g(a,e^a,y)=0,\] or, of the polynomial $g$. Generally, if we replace ``exactly" in condition \ref{item:potential-periodic-interval-3} by ``at most", we call $I$ a \emph{general exponential potential periodic interval set} of the equation $g(a,e^a,y)=0$.
\end{definition}
\begin{remark}\label{remark:why?infty}
The intervals $(a_j,b_j)\;(1\leq j\leq s)$ in a (general) exponential potential periodic interval set cover all the roots of $\lc_{x_2,x_1}(g)$, ensuring that the absolute value $|\lc_{x_2,x_1}(g)|$ is positive outside the set $\bigcup_{j=1}^{s}(a_i,b_i)$. In addition, the unbounded intervals $(a_0,b_0)$ and $(a_{s+1},b_{s+1})$ make sure that the complementary set of $\bigcup_{j=0}^{s+1}(a_i,b_i)$ is bounded.
\end{remark}
 Based on the remark above, the corollary below gives a good $U$ in Proposition \ref{prop:firstbound}.
\begin{corollary}\label{cor:main_case}
Let $g\in\Q[x_1,x_2,x_3]$ containing $x_3$, $[(a_0,b_0),\ldots,(a_{s+1},b_{s+1})]$ be a general exponential potential periodic interval set of $g$, and
\begin{align}
    U:=\R\setminus\bigcup_{j=0}^{s+1}(a_i,b_i).
\end{align}
Then, there exists a positive integer $N$ such that 
$$
\forall x\in \R ~\forall y\in U~(a>N \Rightarrow g(a,e^a,y)\neq 0).
$$
\end{corollary}
\begin{proof}
Recall the definition of a general exponential potential periodic interval set (see Definition \ref{def:potential-periodic-interval}).
We have $a_0=-\infty$, $b_{s+1}=+\infty$, and the set $U$ contains no real root of $\lc_{x_2,x_1}(g)$.
So, $U$ is not only closed but also bounded and $|\lc_{x_2,x_1}(g)|$ is positive over $U$. Hence, $|\lc_{x_2,x_1}(g)|$ has a positive lower bound over the compact set $U$.
Then, the conclusion follows from Proposition \ref{prop:firstbound}.
\end{proof}
\begin{proposition}\label{thm:location-of-the-roots}
Let $g\in\Q[x_1,x_2,x_3]$ be a nonzero polynomial  with $\Delta(x_1,x_2)=\discrim_{x_3}(g)$ nonzero, 
$[(a_0,b_0),\ldots,(a_{s+1},b_{s+1})]$ be a general exponential potential periodic interval set of $g$, $M$ an upper bound of the start-point of the root-functions of $g$ and all real roots of the exponential polynomial
$$\prod_{j=0}^{s}g(a,e^a,b_j)g(a,e^a,a_{j+1}),$$
and $\theta_1(a),\ldots,\theta_r(a)$ be the root-functions of $g$.
Then, for each $i~(1\le i \le r)$, there exists
$j_i~(0\le j_i \le s+1)$ such that $\theta_i(a)\in (a_{j_i},b_{j_i})$ for all $a\in (M,+\infty)$.
\end{proposition}
\begin{proof}
    Note that $g(a,e^a,a_j)~(1\le j \le s+1)$ and $g(a,e^a,b_j)~(0\le j \le s)$ have no roots in $(M,+\infty)$.
Then, for every $i~(1\le i\le r)$ and for any $x\in(M,+\infty)$, $\theta_i(x)\not\in\{b_0,a_1,b_1,\ldots,b_s,a_{s+1}\}$.
Now that $\theta_i(x)$ is continuous, the intermediate value theorem indicates that there exists $j_i~(0\le j_i\le s+1)$ such that $\theta_i((M,+\infty))\subseteq (a_{j_i},b_{j_i})$ or 
there exists $j_i~(0\le j_i\le s)$ such that $\theta_i((M,+\infty))\subseteq(b_{j_i},a_{j_i+1})$.
By Corollary \ref{cor:main_case}, for sufficiently large $a$, $\theta_i(a)\in \cup_{j=0}^{s+1}(a_j,b_j)$. Therefore, $\theta_i((M,+\infty))\subseteq (a_{j_i},b_{j_i})$ for some $0\leq j_i\leq s+1$.
\end{proof}
Note that the general exponential potential periodic interval set in Proposition \ref{thm:location-of-the-roots} is arbitrary (the intervals can be arbitrarily small), the corollary below follows from that proposition:
\begin{corollary}\label{coro:limits}
    $($The same notation used here as in Proposition \ref{thm:location-of-the-roots}.$)$ If $\lambda$ is a real root of $\lc_{x_2,x_1}(g)$ and $\lambda\in(a_{j_i},b_{j_i})$ with $j_i\notin\{0,s+1\}$, then $\lim\limits_{a\rightarrow+\infty}\theta_{i}(a)=\lambda$. If $j_i=0$ or $j_i=s+1$, then $\lim\limits_{a\rightarrow+\infty}\theta_{i}(a)=-\infty$ or $+\infty$, respectively.
\end{corollary}
We combine several results above to conclude the following theorem.
\begin{theorem}\label{thm:farbehave}
    Let $g\in\Q[x_1,x_2,x_3]$ be a nonzero polynomial  with $\Delta(x_1,x_2)=\discrim_{x_3}(g)$ and $g(x_1,x_2,0)$ nonzero, 
$[(a_0,b_0),\ldots,(a_{s+1},b_{s+1})]$ be a general exponential potential periodic interval set of $g$ and $\theta_1,\theta_2,\ldots,\theta_r$ be the root-functions of $g$. For every $i$ $(1\leq i\leq r)$, denote by $j_i$ the index such that $\theta_i(a)\in(a_{j_i},b_{j_i})$ for sufficiently large value $a$.  Then, for all $\delta\in\R_{>0}$ $($except finitely many real positive values$)$, there is an $M>0$ such that
\begin{enumerate}
    \item\label{smallderi} $\forall a>M, 0<\theta_i'(a)<\delta$ or $\;\forall a>M, -\delta<\theta_i'(a)<0$\; $($if\; $0<j_i<s+1)$;
    \item\label{pinf} $\forall a>M, \theta_i'(a)>0$\;$($if\; $j_i=s+1)$; 
    \item\label{minf} $\forall a>M, \theta_i'(a)<0$\;$($if\; $j_i=0)$. 
\end{enumerate}
\end{theorem}
\begin{proof}
By Proposition \ref{prop:almostalllambda}, for all $\delta\in\Q_{>0}$ except some $($finitely many$)$ values, the polynomials $g(x_1,x_2,x_3+\lambda x_1)$, where $\lambda\in\{\pm\delta,0\}$, are nonzero polynomials. On the other hand, $g(x_1,x_2,0)\neq0$ means $\theta(a)\equiv0$ is not a root-function of $g$. Moreover, using the same trick as in the proof of Proposition \ref{prop:almostalllambda}, one claims that $g(x_1,x_2,\lambda x_1)\neq0$ for all $\lambda\in\Q$ except some $($finitely many$)$ values. This indicates that for all $\delta\in\Q_{>0}$ except some $($finitely many$)$ values, the polynomials $g(x_1,x_2,\pm\delta x_1)$ are nonzero and, therefore, the functions $\theta(a)=\pm\delta a$ are not root-functions of $g$. 

Suppose $\delta\in\Q_{>0}$ is a proper number such that non of the polynomials $g(x_1,x_2,x_3\pm\delta x_1)$ and $g(x_1,x_2,\pm\delta x_1)$ is zero. Then, by combining the last paragraph and Corollary \ref{coro:rootfunclargederilambda}, we observe that there is $M>0$ such that for every $1\leq i\leq r$ and every $a>M$, we have $\theta_i'(a)\notin\{\pm\delta,0\}$.  Since $\theta_i'$ is continuous, we have $\theta_i'((M,\infty))\subseteq(-\infty,-\delta)$, $(-\delta,0)$, $(0,\delta)$ or $(\delta,\infty)$.

For those $1\leq i\leq r$ such that $0<j_i<s+1$, we have $\theta(a)\in(a_{j_i},b_{j_i})$ for sufficiently large positive $a$ by Proposition \ref{thm:location-of-the-roots}. The length $\ell=b_{j_i}-a_{j_i}$ is a finite number. Chose $a_1\in\R$ sufficiently large such that $a_2=a_1+\ell/\delta>a_1>M$ and $\theta_i(a_1),\theta_i(a_2)\in(a_{j_i},b_{j_i})$. Then, by Lagrange's mean value theorem, there is $a_3\in(a_1,a_2)$ so that $|\theta_i'(a_3)|=|\frac{\theta_i(a_2)-\theta_i(a_1)}{a_2-a_1}|<\frac{\ell}{\ell/\delta}=\delta$. Note that $a_3>a_1>M$, the inequality $|\theta_i'(a_3)|<\delta$ indicates that $\theta_i'((M,\infty))\subseteq(-\delta,0)$ or $(0,\delta)$. This proves conclusion (\ref{smallderi}).

 We already know $\theta_i'((M,\infty))\subseteq(-\infty,0)$ or $(0,\infty)$ for every $1\leq i\leq r$. This is in particular true for those $i$ such that $j_i=0$ or $s+1$. By Corollary \ref{coro:limits}, we have $\lim\limits_{a\rightarrow\infty}\theta_i(a)=\infty$ if $j_i=s+1$ and $\lim\limits_{a\rightarrow\infty}\theta_i(a)=-\infty$ if $j_i=0$. Noting that it is impossible for an increasing function to tend to $-\infty$ or for a decreasing function to tend to $\infty$, one claims $\theta_i'((M,\infty))\subseteq(-\infty,0)$ when $j_i=0$ and $\theta_i'((M,\infty))\subseteq(0,\infty)$ when $j_i=s+1$. These prove observations  (\ref{pinf}) and (\ref{minf}).
\end{proof}
The following proposition is then clear:
\begin{proposition}\label{命题：有界则导数趋零，无界则导数保号}
   If $\theta$ is a bounded root-function, then \[\lim\limits_{a\rightarrow+\infty}\theta'(a)=0.\] 
   
   If $\theta$ is an  unbounded root-function, then it tends to $\pm\infty$. What's more, $\theta'(a)>0$ for sufficiently large $a$ if $\lim\limits_{a\rightarrow+\infty}\theta(a)=+\infty$, and $\theta'(a)<0$ for sufficiently large $a$ if $\lim\limits_{a\rightarrow+\infty}\theta(a)=-\infty$.
\end{proposition}
\begin{proof}
Suppose $g$ is the irreducible polynomial annihilating $\theta$ according to Proposition \ref{prop:jiexiDaishuzhishuduoxiangshiShiMouyinziGenhanshu}, then $\deg(g,x_3)>0$.

If $\theta\equiv0$ then $\lim\limits_{a\rightarrow+\infty}\theta'(a)=0$ is clear. 

If $\theta$ is not identically zero, then $g(x_1,x_2,0)\neq0$. By Proposition \ref{prop:bukeyuePanbieshiFeiling} and Theorem \ref{thm:farbehave} Conclusion \ref{smallderi}, we have $\lim\limits_{a\rightarrow+\infty}\theta'(a)=0$ if $\theta$ is bounded. If $\theta$ is unbounded, it tends to $\pm\infty$ according to Corollary \ref{coro:limits}. The conclusions follow from Theorem \ref{thm:farbehave} Conclusions \ref{pinf}-\ref{minf}.
\end{proof}

The following proposition gives a sufficient condition for the derivatives of those $\theta_i$'s, with $j_i=s+1$, to be greater than some $\epsilon>0$ for sufficiently large positive $a$:
\begin{proposition}\label{prop:DaoshuJixian}
    Suppose $g\in\mathbb{Q}[x_1,x_2,x_3]$ satisfies $\deg(g,x_3)>0$ and $\discrim(g,x_3)\neq0$. Let $\theta$ be a root-function of $g$ tending to $+\infty$. Define $h(x_1,x_2,x_3)=g(x_1,x_2,x_1x_3)$. Then $\discrim(h,x_3)\neq0$ and if $0\leq \alpha_1<\alpha_2<\cdots<\alpha_k$ are all non-negative roots of $\lc_{x_2,x_1}(h)$, then 
    \[
    \lim_{a\rightarrow+\infty}\theta'(a)\in\{\alpha_1,\alpha_2,\cdots,\alpha_k,+\infty\}.
    \]
    Moreover, the terminal function $\theta(a)/a$ has the same germ with one of the root-functions of $h$, and 
    \[\lim_{a\rightarrow+\infty}\theta(a)/a=\lim_{a\rightarrow+\infty}\theta'(a).\]
\end{proposition}
\begin{proof}
Proposition \ref{prop:x1x3Baochix3wupingfang} shows that $\deg(h,x_3)>0$ and $\discrim(h,x_3)\neq0$. Now we show that $\lim\limits_{a\rightarrow+\infty}\theta'(a)$ is a finite real number or $\pm\infty$: Note that $\theta$ is an analytic algebraic exponential polynomial. By Proposition \ref{prop:derivativeClose}, its derivative $\theta'$ is also an analytic algebraic exponential polynomial. Then, by applying Proposition \ref{prop:jiexiDaishuzhishuduoxiangshiShiMouyinziGenhanshu} to $\theta'$, we claim that $\theta'$ has the same germ with a root-function of some irreducible polynomial. By Corollary \ref{coro:limits}, we conclude that $\lim\limits_{a\rightarrow+\infty}\theta'(a)$ is a finite real number or $\pm\infty$.

    From Proposition \ref{命题：有界则导数趋零，无界则导数保号} we have From$\lim\limits_{a\rightarrow+\infty}\theta'(a)\geq0$. Since $\theta\rightarrow+\infty$ as $a\rightarrow+\infty$, L'Hôpital's rule claims that 
 \begin{equation}\label{equ:luobida}\lim_{a\rightarrow+\infty}\theta(a)/a=\lim_{a\rightarrow+\infty}\theta'(a)\geq 0.\end{equation}
Note that
 \[
h(a,e^a,\theta(a)/a)=g(a,e^a,a\cdot\theta(a)/a)=g(a,e^a,\theta(a))=0.
 \]
 Hence $\theta(a)/a$ has the same germ with one of the root-functions of $h$ by Proposition \ref{prop:lianxuTuigenhanshu}. Combining this with Corollary (\ref{coro:limits}) and condition \ref{equ:luobida} we claim that 
     \[
    \lim_{a\rightarrow+\infty}\theta'(a)\in\{\alpha_1,\alpha_2,\cdots,\alpha_k,+\infty\}.
    \]
\end{proof}

\section{Locating the semi-periodic solutions}\label{sec:semiper}
In this section we locate the semi-periodic solutions of the BETP equations (\ref{equ:ri}) on the graphs of finite many analytic algebraic exponential polynomial tending to infinity. The approach mainly combines the techniques of resultant and the terminal properties of analytic algebraic exponential polynomials. The key concept of semi-periodic solutions is described in Theorems \ref{thm:semiperbunch}, \ref{定理:恰一根} and \ref{定理:右上分布}
\begin{proposition}\label{prop:rootdecom}
    Suppose $r,s$ and $\mathfrak{o}=x_4^2+x_5^2-1\in\Q[x_1,x_2,\ldots,x_5]$ are polynomials. Consider the equations
\begin{equation}\left\{
\begin{array}{rcl}
     r(a,b,e^a,\sin b, \cos b)&=&0  \\
     s(a,b,e^a,\sin b, \cos b)&=&0
\end{array},
\right.
\end{equation}
with $a,b\in\R$ real variables. Define
\begin{equation}\label{方程:pqt}
\left\{\begin{array}{cc}
    p= & \text{res}(\text{res}(r,\mathfrak{o},x_4),\text{res}(s,\mathfrak{o},x_4),x_5)\in\mathbb{Q}[x_1,x_2,x_3] \\
    q= & \text{res}(\text{res}(r,\mathfrak{o},x_4),\text{res}(s,\mathfrak{o},x_4),x_2)\in\mathbb{Q}[x_1,x_3,x_5]\\
    t= & \text{res}(\text{res}(r,\mathfrak{o},x_5),\text{res}(s,\mathfrak{o},x_5),x_2)\in\mathbb{Q}[x_1,x_3,x_4]
\end{array}\right.,
\end{equation}
then there is a polynomial matrix $\mathcal{M}\in(\Q[x_1,x_2,\ldots,x_5])^{3\times3}$ such that \begin{equation}\label{equ:idealrela}
(p,q,t)^T=\mathcal{M}\cdot(r,s,\mathfrak{o})^T.
\end{equation}
Denote by $\{\phi_i\}$, $\{\theta_j\}$ and $\{\vartheta_k\}$ the sets of root-functions of the equations $p(a,y,e^a)=0$, $q(a,e^a,y)=0$ and $t(a,e^a,y)=0$, respectively. Suppose these functions are well-defined over $(N,\infty)$ for some $N>0$. For any triple of root-functions  $(\phi_i,\theta_j,\vartheta_k)$, we have: 

\noindent{}(1) If there is $a_{ijk}>N$ so that  \begin{equation}\label{equ:detFeiling}
\det(\mathcal{M})(a_{ijk},\phi_i(a_{ijk}),e^{a_{ijk}},\vartheta_k(a_{ijk}),\theta_j(a_{ijk}))\neq0,
\end{equation}
then there is $N_{ijk}>N$ such that $\forall a>N_{ijk},$
\begin{equation}\label{equ:detnonzero}
    \det(\mathcal{M})(a,\phi_i(a),e^a,\vartheta_k(a),\theta_j(a))\neq0.
\end{equation}
Moreover, a point $(a,b)=(a_*,\phi_i(a_*))$ with $a_*>N_{ijk}$ is a solution to equations (\ref{equ:ri}) whenever 
\begin{equation}\label{equ:sincosgood}
\left\{\begin{array}{rcl}
     \vartheta_k(a_*)&=& \sin\phi_i(a_*) \\
     \theta_j(a_*)&=& \cos\phi_i(a_*)
\end{array}\right..
\end{equation}
\noindent{}(2) If there is $a_{ijk}>N$ so that
\begin{equation}\label{equ:rsoFeiling}\left\{
\begin{array}{rcl}
     r(a_{ijk},\phi_i(a_{ijk}),e^{a_{ijk}},\vartheta_k(a_{ijk}),\theta_j(a_{ijk}))&\neq&0,\text{\;\;or}\\
     s(a_{ijk},\phi_i(a_{ijk}),e^{a_{ijk}},\vartheta_k(a_{ijk}),\theta_j(a_{ijk}))&\neq&0\text{,\;\;or}\\
     (\vartheta_k(a_{ijk}))^2+(\theta_j(a_{ijk}))^2-1&\neq&0,
\end{array}\right.\end{equation}
then there is $N_{ijk}>N$ such that any point $(a_*,\phi_i(a_*))$ with $a>N_{ijk}$ is not a solution to equations (\ref{equ:ri}) whenever equations (\ref{equ:sincosgood}) hold.
\vspace{1.5mm}

\noindent{}(3) If for every $\;a>N$,
 \begin{equation}\label{equ:detHengling}
\det(\mathcal{M})(a,\phi_i(a),e^{a},\vartheta_k(a),\theta_j(a))=0,
\end{equation}
and
\begin{equation}\label{equ:rso=0}
\left\{
\begin{array}{rcl}
     r(a,\phi_i(a),e^a,\vartheta_k(a),\theta_j(a))&=&0\\
     s(a,\phi_i(a),e^a,\vartheta_k(a),\theta_j(a))&=&0\\
     (\vartheta_k(a))^2+(\theta_j(a))^2-1&=&0
\end{array}\right.,
\end{equation}
then a point $(a,b)=(a_*,\phi_i(a_*))$ with $a_*>N$ is a solution to equations (\ref{equ:ri}) whenever equations (\ref{equ:sincosgood}) hold.
\end{proposition}
\begin{proof}
By the property of the resultant, the polynomials $p,q,t$ are all in the idea generated by $r, s$ and $\mathfrak{o}$ in the ring $\Q[x_1,x_2,\ldots,x_5]$. Thus the matrix $\mathcal{M}$ exists (if fact in can be computed).
\vspace{1.5mm}

\noindent{}\textcolor{blue}{Proof of case (1)}: By Proposition \ref{prop:isfield}, the function 
\[\det(\mathcal{M})(a,\phi_i(a),e^a,\vartheta_j(a),\theta_k(a))\] is an analytic algebraic exponential polynomial. Therefore, the existence of the number $N_{ijk}$ is guaranteed by Proposition \ref{prop:rootfunclarge} and condition (\ref{equ:detFeiling}).

Combining equation (\ref{equ:idealrela}) and inequality (\ref{equ:detnonzero}), one observes that for all $a>N_{ijk}$,
\begin{equation}\label{equ:zeroequizero}
    \left\{\begin{array}{rcl}
     p(a,\phi_i(a),e^a)&=&0\\
     q(a,e^a,\theta_j(a))&=&0\\
     t(a,e^a,\vartheta_k(a))&=&0
\end{array}\right.\Longleftrightarrow\left\{
\begin{array}{rcl}
     r(a,\phi_i(a),e^a,\vartheta_k(a),\theta_j(a))&=&0\\
     s(a,\phi_i(a),e^a,\vartheta_k(a),\theta_j(a))&=&0\\
     (\vartheta_k(a))^2+(\theta_j(a))^2&=&1
\end{array}\right..
\end{equation}

Note that $\phi_i$, $\theta_j$ and $\vartheta_k$ are root-functions of $p,\,q$ and $t$, the equations in the left column hold for all $a>N_{ijk}$. Therefore, the equations in the right column also hold for every $a>N_{ijk}$.  Thus, if $(a_*,\phi_i(a_*))$ satisfies condition (\ref{equ:sincosgood}), it is a solution to equations (\ref{equ:ri}).
\vspace{1.5mm}

\noindent{}\textcolor{blue}{Proof of case (3)}:

Now that equations (\ref{equ:rso=0}) hold for every $a>N$, the conclusion is clear.
\vspace{1.5mm}

\noindent{}\textcolor{blue}{Proof of case (2)}: Without loss of generality, we suppose that
\begin{equation}\label{equ:raijkBudengLing}
r(a_{ijk},\phi_i(a_{ijk}),e^{a_{ijk}},\vartheta_k(a_{ijk}),\theta_j(a_{ijk}))\neq0.
\end{equation}
Again, by Proposition \ref{prop:isfield}, any one of the functions on the left-hand sides of equations (\ref{equ:rso=0}), say the function
\[r(a,\phi_i(a),e^a,\vartheta_k(a),\theta_j(a)),\]
is an analytic algebraic exponential polynomial. Therefore, by Proposition \ref{prop:rootfunclarge} and condition (\ref{equ:raijkBudengLing}), there is a number $N_{ijk}>N$ such that 
\[r(a,\phi_i(a),e^a,\vartheta_k(a),\theta_j(a))\neq0,\]
for any $a>N_{ijk}$. Hence, any point $(a_*,\phi_i(a_*))$ with $a_*>N_{ijk}$ is not a solution to equations (\ref{equ:ri})  when equations (\ref{equ:sincosgood}) hold.
\end{proof}
\begin{remark}
    Conditions (\ref{equ:detFeiling}) and (\ref{equ:rsoFeiling}) can be verified by interval arithmetic techniques, while conditions (\ref{equ:detHengling}) and (\ref{equ:rso=0}) cannot. The interval arithmetic techniques only ``suggest" that the latter two conditions hold when they fail to verify the first two.
\end{remark}
\begin{definition}\label{定义:根线组合情形分类}
    In Proposition \ref{prop:rootdecom}, we say the triple of root-functions $(\phi_i,\theta_j,\vartheta_k)$ is of the \emph{non-degenerate/degenerate case} if it is satisfies case (1)/(3) therein. We say it is of the \emph{no-solution case} if it satisfies case (2).
\end{definition}
\begin{definition}\label{定义:三分界}
    For every triple $(\phi_i,\theta_j,\vartheta_k)$ that is of the nondegenerate case or the no-solution case, the boundary $N_{ijk}$ is given in Proposition \ref{prop:rootdecom}. For every $(\phi_i,\theta_j,\vartheta_k)$ that is of the degenerate case, we define $N_{ijk}=N$, with $N$ being the common stat-point of those root-functions in the sets $\{\phi_i\}$, $\{\theta_j\}$ and $\{\vartheta_k\}$. Then, the number $N^*=\max\limits_{i,j,k}N_{ijk}$ is called \emph{the tripartite boundary} of equations (\ref{equ:ri}).
\end{definition}
\begin{remark}\label{评论:退化非退化满足}
    From the proof of Proposition \ref{prop:rootdecom}, a triple $(\phi_i,\theta_j,\vartheta_k)$ that is of degenerate or non-degenerate case satisfies equations (\ref{equ:rso=0}) for all $a>N^*$.
\end{remark}
Then following corollary is straightforward:
\begin{corollary}\label{引理:根按ijk分解}
We use the same notation as in Proposition \ref{prop:rootdecom}. Set $N^*=\max\limits_{i,j,k}N_{ijk}$ to be the tripartite boundary, then a point $(a_*,b_*)$ with $a_*>N^*$ is a solution to equations (\ref{equ:ri}) if and only if there is a triple $(\phi_i,\theta_j,\vartheta_k)$ of root-functions, which is of the degenerate or the non-degenerate case, such that
\begin{equation}\label{方程:跟函数分别等于b正余弦}
\left\{\begin{array}{rcl}
     \phi_i(a_*) &=&b_* \\
     \vartheta_k(a_*)&=& \sin\phi_i(a_*) \\
     \theta_j(a_*)&=& \cos\phi_i(a_*)
\end{array}\right..
\end{equation}
\end{corollary}
\begin{proof}
    The ``if" part of the proof is clear, since equations (\ref{equ:rso=0}) hold for every $a=a_*>N^*$ in both the degenerate and the non-degenerate cases.

    Now we prove the ``only if" part: Suppose $(a,b)=(a_*,b_*)$ is a solution to equations (\ref{equ:ri}). From equation (\ref{equ:idealrela}), we conclude that it is also a solution to the equations
    \[
    \left\{
    \begin{array}{rcl}
         p(a,b,e^a)&=& 0 \\
          q(a,e^a,\cos b)&=&0  \\
         t(a,e^a,\sin b)&=&0
    \end{array}
    \right..
    \]
    According to Theorem \ref{theorem:rootfunc}, there are root-functions $\phi_i,\theta_j$ and $\vartheta_k$ of these equations, such that 
    \[
\left\{\begin{array}{rcl}
     \phi_i(a_*)&=&b_*  \\
     \vartheta_k(a_*)&=& \sin b_* \\
     \theta_j(a_*)&=& \cos b_*
\end{array}\right.,
\]which are just equations (\ref{方程:跟函数分别等于b正余弦}). Now we have a solution $(a_*,\phi_i(a_*))$ to equations (\ref{equ:ri}) with equations (\ref{equ:sincosgood}) hold, contradicting the conclusion of the second part of Proposition \ref{prop:rootdecom}. This means that the triple $(\phi_i,\theta_j,\vartheta_k)$ is of the degenerate or the non-degenerate case.
\end{proof}
\begin{remark}
    For every triple $(\phi_i,\theta_j,\vartheta_k)$, we use interval arithmetic techniques to recognize the non-degenerate case and the no-solution case, then we find those $(a_*,b_*)$ satisfying equations (\ref{方程:跟函数分别等于b正余弦}) to obtain the solutions with $a_*>N^*$.
\end{remark}

The proposition below  shows how to decide the condition (\ref{方程:跟函数分别等于b正余弦}) for a specific triple $(\phi_i,\theta_j,\vartheta_k)$.
\begin{proposition}\label{prop:2+sign}
    For any triple $(\phi_i,\theta_j,\vartheta_k)$ of the degenerate or the non-degenerate case with $\phi_i$ not a constant, there is $N^s>N^*$ such that for all $a>N^s$, the signs $\mathfrak{b}=\text{\emph{sgn}}(\phi'_i(a))\neq0$ and $\mathfrak{s}=\text{\emph{sgn}}(\vartheta_k(a))$ are both invariant. Then, for any $a_*>N^s$, equations (\ref{方程:跟函数分别等于b正余弦}) hold iff its first and last equations and the condition 
    \begin{equation}\label{equ:sgncondi}
\text{\emph{{sgn}}}\big(\frac{\mathrm{d}\cos\phi_i(a)}{\mathrm{d}a}\big|_{a=a_*}\big)=-\mathfrak{bs} 
    \end{equation}hold.
\end{proposition}
\begin{proof}
Note that $\phi_i'$ is an analytic algebraic exponential polynomial (Proposition \ref{prop:derivativeClose}), so is $\vartheta_k$. Hence $N^s$ exists according to Proposition \ref{prop:rootfunclarge} and Proposition \ref{prop:rootfunclargederi}, and the continuity of $\phi_i'$ and $\vartheta_k$.

Suppose that $\text{{sgn}}\big(\frac{\mathrm{d}\cos\phi_i(a)}{\mathrm{d}a}\big|_{a=a_*}\big)=-\mathfrak{bs}$ and $a_*>N^s$ is such that $\theta_j(a_*)=\cos\phi_i(a_*)$. Since the triple $(\phi_i,\theta_j,\vartheta_k)$ is of the degenerate or the non-degenerate case, equations (\ref{equ:rso=0}) hold for any $a=a_*>N^s>N^*$. We  have
\begin{equation}\label{equ:vartheta=pmsin}
\begin{array}{rcl}
    \vartheta_k(a_*) &= &\pm\sqrt{1-(\theta_j(a_*))^2}  \\
     & =&\pm\sqrt{1-(\cos\phi_i(a_*))^2}\\
     & =&\pm\sin\phi_i(a_*).
\end{array}
\end{equation}
On the other hand, 
\begin{equation}\label{方程:cosphi求导}
    \frac{\mathrm{d}\cos\phi_i(a)}{\mathrm{d}a}\big|_{a=a_*}=-\phi_i'(a_*)\sin\phi_i(a_*)
\end{equation}
implying that $\text{{sgn}}(\sin\phi_i(a_*))=-\,\text{{sgn}}\big(\phi_i'(a_*)\frac{\mathrm{d}\cos\phi_i(a)}{\mathrm{d}a}\big|_{a=a_*}\big)=-(\mathfrak{b}\cdot(-\mathfrak{b}\mathfrak{s}))=\mathfrak{s}=\text{{sgn}}(\vartheta_k(a_*))$. Therefore, $\vartheta_k(a_*)=\sin\phi_i(a_*)$ follows from equation (\ref{equ:vartheta=pmsin}). 

The reverse implication is clear according to equation (\ref{方程:cosphi求导}).
\end{proof}
\begin{remark}\label{rem:intersection}
    In equations (\ref{方程:跟函数分别等于b正余弦}), there are 3 formulae. Proposition \ref{prop:2+sign} reduces them to the sign-condition (\ref{equ:sgncondi}) and other 2 formulae, meaning two curves intersect each other in geometry, which is much simpler.
\end{remark}

\begin{theorem}\label{thm:semiperbunch}
Set $(\phi_i,\theta_j,\vartheta_k)$ to be a triple of degenerate of non-degenerate case. Suppose that $\lim\limits_{a\rightarrow+\infty}\phi_i(a)=+\infty$ and the number $N^s$ is as in Proposition \ref{prop:2+sign}, then $\mathfrak{b}=\emph{\text{sgn}}(\phi'_i((N^s,\infty)))=1$ and $\{\theta_j(a),\vartheta_k(a)\}\subset[-1,1]$ for every $a>N^s$. Moreover,
 
 (1) if $\mathfrak{s}=1$, and $M> N^s$ is any maximal point of the function $y=\cos \phi_i(a)$, $m>M$ is the smallest minimal point of the function $y=\cos \phi_i(a)$ exceeding $M$, then there are at least one but at most finitely many $a_*\in(M,m)$ satisfying equations (\ref{equ:sincosgood}). Thus $(a_*,\phi_i(a_*))$ is a solution to equations (\ref{equ:ri}). Moreover, there is no such $a_*$ in any closed increasing intervals of $y=\cos \phi_i(a)$;

    (2) if $\mathfrak{s}=-1$, and $m> N^s$ is any minimal point of the function $y=\cos\phi_i(a)$, $M>m$ is the smallest maximal point of the function $y=\cos\phi_i(a)$ exceeding $m$, then there are at least one but at most finitely many $a_*\in(m,M)$ satisfying equations (\ref{equ:sincosgood}). Thus $(a_*,\phi_i(a_*))$ is a solution to equations (\ref{equ:ri}). Moreover, there is no such $a_*$ in any closed decreasing intervals of $y=\cos \phi_i(a)$;

    (3) if $\mathfrak{s}=0$, those maximal (minimal) points of the function $y=\cos\phi_i(a)$ exceeding $N^s$ are the only values of $a_*>N^s$ satisfying equations (\ref{equ:sincosgood}) when $\theta_i(a)\equiv1$ $(\theta_i(a)\equiv-1)$.

In all cases, there are countably many real solutions to equations (\ref{equ:ri}) on the graph of the function $\phi_i$ restricted to $(N^s,\infty)$.
\end{theorem}
\begin{proof}
According to Proposition \ref{prop:2+sign} and Remark \ref{rem:intersection}, we consider the intersection of the curves $y=\theta_j(a)$ and $y=\cos\phi_i(a)$. 

The sign of $\phi'_i(a)$ is invariant and nonzero for all $a>N^s$ by Proposition \ref{prop:2+sign}. Now that $\lim\limits_{a\rightarrow+\infty}\phi_i(a)=+\infty$, $\mathfrak{b}=\text{sgn}(\phi'_i(a))=1$ for all $a>N^s$. Therefore, $\phi_i(a)$ is monotonously increasing in $(N^s,\infty)$. The curve of $y=\cos\phi_i(a)$ is similar to the one of the cosine function: oscillating from $-1$ to 1 and from 1 back to -1 again and again for countably many times. Since $N^s>N^*$, $\theta_j^2(a)+\vartheta_k^2(a)=1$ for all $a>N^s$. Thus $\theta_j((N^s,\infty))\subseteq[-1,1]$ and $\vartheta_j((N^s,\infty))\subseteq[-1,1]$. Every time the curve of the function $y=\cos\phi_i(a)$ increases from $-1$ to 1, it intersects the curve of the function $y=\theta_j(a)$ at least once. The same thing happens when the curve of the function $y=\cos\phi_i(a)$ decreases from 1 back to $-1$.

Suppose $\mathfrak{s}=\text{sgn}(\vartheta_k((N^s,\infty)))=1$, then equation (\ref{equ:vartheta=pmsin}) indicates that  the function $\theta_j(a)\notin\{-1,1\}$ for all $a>N^s$. Note that $-\mathfrak{b}\mathfrak{s}=-1$, by Proposition \ref{prop:2+sign} and condition (\ref{equ:sgncondi}), whenever the curve of the function $y=\cos\phi_i(a)$ decreases from $1$ back to $-1$ and intersects the curve of the function $y=\theta_j(a)$ at some $a=a_*$ (which satisfies $\cos\phi_i(a_*)\not\in\{\pm1\}$), the point $(a_*,b_*)=(a_*,\phi_i(a_*))$ satisfies equations (\ref{方程:跟函数分别等于b正余弦}). Thus $(a_*,b_*)$ is a solution to equations (\ref{equ:ri}) by Proposition \ref{prop:rootdecom}. While the intersections obtained when the curve of the function $y=\cos\phi_i(a)$ increases from $-1$ to 1 do not satisfy equations (\ref{方程:跟函数分别等于b正余弦}). Thus those $(a_*,b_*)$ are the only solutions to equations (\ref{方程:跟函数分别等于b正余弦}) if  $a_*>N^s$ is required.

Suppose $(M,m)$ is a decreasing interval of $y=\cos\phi_i(a)$ with $M$ a maximal point ($\cos\phi_i(M)=1$) and $m$ a minimal point ($\cos\phi_i(m)=-1$). We show that there are finitely many $a_*\in(M,m)$ satisfying equations (\ref{equ:sincosgood}). If there are infinitely many such points, then they are zeros of the analytic function $y=\cos\phi_i(a)-\theta_j(a)$ and they have a limit point in $[M,m]$. This means $\cos\phi_i(a)\equiv\theta_j(a)$. But $\lim\limits_{a\rightarrow+\infty}\theta_j(a)$ exists while $\lim\limits_{a\rightarrow+\infty}(\cos\phi_i(a))$ does not, which is a contradiction.

The situation when $\mathfrak{s}=-1$ is opposite: the intersections obtained while the curve of the function $y=\cos\phi_i(a)$ increases from $-1$ to 1 yield at least countably many solutions to equations (\ref{equ:ri}).

If $\mathfrak{s}=0$, $\vartheta_k(a)\equiv0$. Then $\theta_j(a)\equiv1$ or $\theta_j(a)\equiv-1$ by equation (\ref{equ:vartheta=pmsin}). The $a$-coordinates of the intersections of the curves $y=\cos\phi_i(a)$ and $y=\theta_j(a)$ are those maxima or those minima of the function  $y=\cos\phi_i(a)$, respectively. In this case the condition (\ref{equ:sgncondi}) holds. Thus those maxima or minima yield countably many solutions to equations (\ref{equ:ri}).
\end{proof}
\begin{theorem}\label{定理:导数正极限推出恰一根}
    We use the same notation and assumptions as in Theorem \ref{thm:semiperbunch}. Suppose in addition that the root-function $\phi_i(a)/a$ of the equation $p(a,ay,e^a)=0$ has a positive limit (either finite or not), then there is $N^{\times}>N^s$ such that: 

    (1) If $\mathfrak{s}=1$, and $M> N^{\times}$ is any maximal point of the function $y=\cos \phi_i(a)$, $m>M$ is the smallest minimal point of the function $y=\cos \phi_i(a)$ exceeding $M$, then there is a unique $a_*\in(M,m)$ satisfying equations (\ref{equ:sincosgood}). Thus $(a_*,\phi_i(a_*))$ is a solution to equations (\ref{equ:ri}).

    (2) If $\mathfrak{s}=-1$, and $m> N^{\times}$ is any minimal point of the function $y=\cos\phi_i(a)$, $M>m$ is the smallest maximal point of the function $y=\cos\phi_i(a)$ exceeding $m$, then there is a unique $a_*\in(m,M)$ satisfying equations (\ref{equ:sincosgood}). Thus $(a_*,\phi_i(a_*))$ is a solution to equations (\ref{equ:ri}).

    (3) If $\mathfrak{s}=0$, those maximal (minimal) points of the function $y=\cos\phi_i(a)$ exceeding $N^s$ are the only values of $a_*$ satisfying equations (\ref{equ:sincosgood}) when $\theta_i(a)\equiv1$ $(\theta_i(a)\equiv-1)$.

    In all cases, there are exactly countably many solutions to equations (\ref{equ:ri}).
\end{theorem}
\begin{proof}
    The case when $\mathfrak{s}=0$ is already proven in Theorem \ref{thm:semiperbunch}. In the following we prove the case when $\mathfrak{s}=1$. The proof of the case when $\mathfrak{s}=-1$ is similar. The proof of the case when $\mathfrak{s}=1$ is split into two sub-cases:

    \textcolor{blue}{Case 1}: When $\lim\limits_{a\rightarrow+\infty}\theta_j(a)\in(-1,1)$.

    Now that $\lim\limits_{a\rightarrow+\infty}\phi_i(a)/a>0$, by Proposition \ref{prop:DaoshuJixian}, there is $\delta>0$ and $N^s_1>N^s$ such that $\phi_i'(a)>\delta$ for all $a>N^s_1$. Since $\lim\limits_{a\rightarrow+\infty}\theta_j(a)\triangleq\lambda\in(-1,1)$, there is $\varepsilon>0$ so that $(\lambda-\varepsilon,\lambda+\varepsilon)\subset(-1,1)$ and $d=\min\{1-(\lambda+\varepsilon),(\lambda-\varepsilon)-(-1)\}>0$. There is $N^s_2>N^s$ so that $\theta_j(a)\in(\lambda-\varepsilon,\lambda+\varepsilon)$ for all $a>N^s_2$. By Proposition \ref{命题：有界则导数趋零，无界则导数保号}, we have $\theta_j'(a)\rightarrow0$. Hence there is $N^s_3>N^s$ such that $|\theta_j'(a)|<d\delta$ for all $a>N^s_3$. We set $N^{\times}=\max\{N^s_1,N^s_2,N^s_3\}$.

    The function $\cos\phi_i(a)$ is strictly decreasing in the interval $(M,m)$. There is a unique $M_+\in(M,m)$ such that $\cos\phi_i(M_+)=\lambda+\varepsilon$ and a unique $m_-\in(M,m)$ such that $\cos\phi_i(m_-)=\lambda-\varepsilon$. Moreover, we have $M<M_+<m_-<m$. In the intervals $[M,M_+]$ and $[m_-,m]$ there is no point $a_*$ satisfying equations (\ref{equ:sincosgood}) since on those intervals we have $\theta_j(a)\in(\lambda-\varepsilon,\lambda+\varepsilon)$ and $\cos\phi_i(a)\notin(\lambda-\varepsilon,\lambda+\varepsilon)$. Now it suffices to show that there is a unique $a_*\in(M_+,m_-)$ satisfying equations (\ref{equ:sincosgood}).

    For all $a\in(M_+,m_-)$ we have $\cos\phi_i(a)\in(\lambda-\varepsilon,\lambda+\varepsilon)$. Therefore, $1-\cos\phi_i(a)>d$ and $\cos\phi_i(a)-(-1)>d$. These gives \[\sin^2\phi_i(a)=1-\cos^2\phi_i(a)>d^2,\]
    meaning that $|\sin\phi_i(a)|>d$. Therefore,
    \[
    |(\cos\phi_i(a))'|=|\sin\phi_i(a)||\phi'_i(a)|>d\delta
    \]
    for all $a\in(M_+,m_-)$. Now that $|\theta_j(a)|<d\delta$, the derivative of the function $\cos\phi_i(a)-\theta_j(a)$ is negative on $(M_+,m_-)$. Thus there is a unique $a_*\in(M_+,m_-)$ satisfying equations (\ref{equ:sincosgood}).
    
 \textcolor{blue}{Case 2}: When $\lim\limits_{a\rightarrow+\infty}\theta_j(a)\in\{-1,1\}$.

In this case $\theta_j$ is not identically zero. By Proposition \ref{prop:rootfunclarge}, there is $N^s_1>N^s$ so that $\text{sgn}(\theta_j(a))$ is invariant for all $a>N^s_1$. We denote it by $\mathfrak{{c}}\neq0$. Now that $\mathfrak{b}=1$, one verifies that for all $a>N^s_1$, the following two sets of equations are equivalent:
 \begin{equation}\label{方程:正余弦转化}
 \left\{
 \begin{array}{rcl}
      \cos\phi_i(a)&=&\theta_j(a)  \\
      \text{sgn}((\cos\phi_i(a))')&=&-\mathfrak{bs} 
 \end{array}
 \right.
 \Longleftrightarrow
  \left\{
 \begin{array}{rcl}
      \sin\phi_i(a)&=&\vartheta_k(a)  \\
      \text{sgn}((\sin\phi_i(a))')&=&\mathfrak{bc}
 \end{array}
 \right..
 \end{equation}
 These are both equivalent to equations (\ref{equ:sincosgood}). Thus we can consider the intersections of the functions $y=\sin\phi_i(a)$ and $y=\vartheta_k(a)$ instead.

 Note that we have $\vartheta_k(a)\rightarrow0\in(-1,1)$. Using the trick in the proof of Case 1, one can prove that if $\mathfrak{c}=1$ ($\mathfrak{c}=-1$), then there is $N^s_2>N^s_1$ such that for any increasing (decreasing) interval $(v,u)\subset(N^s_2,\infty)$ ($(u,v)\subset(N^s_2,\infty)$) of the function $y=\sin\phi_i(a)$, with $u$ a maximal point and $v$ a minimal point, there is a unique $a=a_*\in(v,u)$ ($a=a_*\in(u,v)$) satisfying the conditions in (\ref{方程:正余弦转化}).

 Denote by $v_1<v_2<v_3<\cdots$ all those successive minimal point of $y=\sin\phi_i(a)$ greater than $N^s_2$, and by $u_1<u_2<u_3<\cdots$ all those maximal point of  $y=\sin\phi_i(a)$ greater than $v_1$. Denote by $M_1<M_2<M_3<\cdots$ and by $m_1<m_2<m_3<\cdots$ all those successive maximal and minimal points of $y=\cos\phi_i(a)$, respectively, such that $M_1>v_1$  and $m_1>u_1$. Then
 \[
v_1<M_1<u_1<m_1<v_2<M_2<u_2<m_2<\cdots, 
 \]
 and $M_1<m_1<M_2<m_2<\dots$ are zeros of $y=\sin\phi_i(a)$.

 Taking $N^{\times}=v_1$, then all the decreasing intervals of $y=\cos\phi_i(a)$ contained in $(N^{\times},\infty)$ are $(M_1,m_1), (M_2,m_2),\cdots$. We need to show that there is a unique $a=a_*\in(M_\iota,m_\iota)$ satisfying conditions in (\ref{方程:正余弦转化}) for $\iota\geq1$. 
 
 Form above we know that, if $\mathfrak{c}=1$ ($\mathfrak{c}=-1$), then there is a unique $a=a_*\in(v_\iota,u_\iota)$ ($a=a_*\in(u_\iota,v_{\iota+1})$) satisfying the conditions in (\ref{方程:正余弦转化}). Since $\mathfrak{s}=1$, $\sin\phi_i(a_*)>0$. Hence $a_*\in(M_\iota,u_\iota)$ ($a_*\in(u_\iota,m_\iota)$). On the other hand, there is no $a\in[u_\iota,m_\iota)$ ($a\in(M_\iota,u_\iota]$) satisfying the conditions in (\ref{方程:正余弦转化}), since $(\sin\phi_i(a))'\leq0$ ($\geq0$) in that interval but conditions in (\ref{方程:正余弦转化}) require that \[\text{sgn}(\sin\phi_i(a))'=\mathfrak{bc}>0\;\; (<0).\] 
\end{proof}
The condition that the root-function $\phi_i(a)/a$ has a positive limit can be decided according to Corollary \ref{coro:limits}. The following theorem claims that this condition can be removed:
\begin{theorem}\label{定理:恰一根}
    Using the same notation and assumptions as in Theorem \ref{thm:semiperbunch}, we claim that there is $N^{\times}>N^s$ such that: 

    (1) If $\mathfrak{s}=1$, and $M> N^{\times}$ is any maximal point of the function $y=\cos \phi_i(a)$, $m>M$ is the smallest minimal point of the function $y=\cos \phi_i(a)$ exceeding $M$, then there is a unique $a_*\in(M,m)$ satisfying equations (\ref{equ:sincosgood}). Thus $(a_*,\phi_i(a_*))$ is a solution to equations (\ref{equ:ri}).

    (2) If $\mathfrak{s}=-1$, and $m> N^{\times}$ is any minimal point of the function $y=\cos\phi_i(a)$, $M>m$ is the smallest maximal point of the function $y=\cos\phi_i(a)$ exceeding $m$, then there is a unique $a_*\in(m,M)$ satisfying equations (\ref{equ:sincosgood}). Thus $(a_*,\phi_i(a_*))$ is a solution to equations (\ref{equ:ri}).

    (3) If $\mathfrak{s}=0$, those maximal (minimal) points of the function $y=\cos\phi_i(a)$ exceeding $N^s$ are the only values of $a_*$ satisfying equations (\ref{equ:sincosgood}) when $\theta_i(a)\equiv1$ $(\theta_i(a)\equiv-1)$.

    In all cases, there are exactly countably many solutions to equations (\ref{equ:ri}).
\end{theorem}
\begin{proof}
    The case when $\mathfrak{s}=0$ is already proven in Theorem \ref{thm:semiperbunch}. In the following we prove the case when $\mathfrak{s}=1$. The proof of the case when $\mathfrak{s}=-1$ is similar. The proof of the case when $\mathfrak{s}=1$ is split into two sub-cases:

    \textcolor{blue}{Case 1}: When $\lim\limits_{a\rightarrow+\infty}\theta_j(a)\in(-1,1)$.

    
    Since $\lim\limits_{a\rightarrow+\infty}\theta_j(a)\triangleq\lambda\in(-1,1)$, there is $\varepsilon>0$ so that $(\lambda-\varepsilon,\lambda+\varepsilon)\subset(-1,1)$ and $d=\min\{1-(\lambda+\varepsilon),(\lambda-\varepsilon)-(-1)\}>0$. There is $N^s_1>N^s$ so that $\theta_j(a)\in(\lambda-\varepsilon,\lambda+\varepsilon)$ for all $a>N^s_1$. Let $\delta <d$ be a positive rational number, then the function $y=\delta\phi_i(a)+\theta_j(a)$ is an analytic algebraic exponential polynomial by Proposition \ref{prop:isfield}. Now that $\phi_i(a)$ tends to $+\infty$ and $\theta_j(a)$ is bounded, we have $\delta\phi_i(a)+\theta_j(a)\rightarrow+\infty$. Combining Proposition \ref{prop:rootfunclargederi} with the continuity of the derivative  $\delta\phi_i'(a)+\theta_j'(a)$, we claim that  there is $N^s_2>N^s$ such that 
    \[
    \forall a>N^s_2,\;\;\delta\phi_i'(a)+\theta_j'(a)>0,\;\text{or}
    \]
     \[
    \forall a>N^s_2,\;\;\delta\phi_i'(a)+\theta_j'(a)<0.
    \]
    But the latter is impossible since $\delta\phi_i(a)+\theta_j(a)\rightarrow+\infty$. We then set $N^{\times}=\max\{N^s_1,N^s_2\}$.
    

    The function $\cos\phi_i(a)$ is strictly decreasing in the interval $(M,m)$. There is a unique $M_+\in(M,m)$ such that $\cos\phi_i(M_+)=\lambda+\varepsilon$ and a unique $m_-\in(M,m)$ such that $\cos\phi_i(m_-)=\lambda-\varepsilon$. Moreover, we have $M<M_+<m_-<m$. In the intervals $[M,M_+]$ and $[m_-,m]$ there is no point $a_*$ satisfying equations (\ref{equ:sincosgood}) since on those intervals we have $\theta_j(a)\in(\lambda-\varepsilon,\lambda+\varepsilon)$ and $\cos\phi_i(a)\notin(\lambda-\varepsilon,\lambda+\varepsilon)$. Now it suffices to show that there is a unique $a_*\in(M_+,m_-)$ satisfying equations (\ref{equ:sincosgood}).

    For all $a\in(M_+,m_-)$ we have $\cos\phi_i(a)\in(\lambda-\varepsilon,\lambda+\varepsilon)$. Therefore, $1-\cos\phi_i(a)>d$ and $\cos\phi_i(a)-(-1)>d$. These gives \[\sin^2\phi_i(a)=1-\cos^2\phi_i(a)>d^2,\]
    meaning that $|\sin\phi_i(a)|>d>\delta$. Noting that the derivative \[(\cos\phi_i(a))'=-\phi'_i(a)\sin\phi_i(a)<0\] in the interval $(M,m)$ and $-\phi_i'(a)<0$, we have $\sin\phi_i(a)>0$ and thus $\sin\phi_i(a)>\delta$ for all $a\in(M_+,m_-)$. Therefore,
    \[
\begin{array}{rcl}
     (\cos\phi_i(a)-\theta_j(a))'&=&
    -\phi'_i(a)\sin\phi_i(a)-\theta_j'(a)  \\
     &<&-\delta\phi'_i(a)-\theta_j'(a)\\
     &<&0
\end{array}
    \]
    for all $a\in(M_+,m_-)$. This means the function $\cos\phi_i(a)-\theta_j(a)$ is monotonically decreasing on $(M_+,m_-)$. Thus there is a unique $a_*\in(M_+,m_-)$ satisfying equations (\ref{equ:sincosgood}).
    
 \textcolor{blue}{Case 2}: When $\lim\limits_{a\rightarrow+\infty}\theta_j(a)\in\{-1,1\}$.

This part is the same as Case 2 in the proof of Theorem \ref{定理:导数正极限推出恰一根}.
\end{proof}

\begin{remark}\label{评论：为什么叫做半周期解}
    Theorem \ref{thm:semiperbunch} yields countably many solutions \[(a_n,b_n)=(a_*^{(n)},\phi_i(a_*^{(n)}))\] with $n=1,2,3,\cdots$ and $a_*^{(1)}<a_*^{(2)}<a_*^{(3)}<\cdots$. Note that $\lim\limits_{n\rightarrow+\infty}a_*^{(n)}=+\infty,$ we have \[\lim\limits_{n\rightarrow+\infty}\cos b_{n}=\lim\limits_{n\rightarrow+\infty}\cos\phi_i(a_*^{(n)})=\lim\limits_{n\rightarrow+\infty}\theta_j(a_*^{(n)})=\lim\limits_{a\rightarrow+\infty}\theta_j(a).\]
The last limit exists according to Corollary \ref{coro:limits}. Suppose $\lim\limits_{a\rightarrow+\infty}\theta_j(a)=\lambda$ and $\lambda\in[\lambda_-,\lambda_+]=[\max\{\lambda-\epsilon,-1\},\min\{\lambda+\epsilon,1\}]$ for $\epsilon>0$, then for sufficiently large $n$, $\cos b_n\in[\lambda_-,\lambda_+]$. That is, for each $n$ sufficiently large, there is $\ell_n\in\Z_{>0}$ such that 
\[
\left\{\begin{array}{rcl}
     b_n&\in&[\arccos \lambda_+,\arccos \lambda_-]+2\ell_n\pi\text{, or} \\
     b_n&\in&[-\arccos \lambda_-,-\arccos \lambda_+]+2(\ell_n+1)\pi,
\end{array}
\right.
\]
with $\{\ell_n\}$ a non-decreasing integer sequence satisfying $\ell_{n+1}-\ell_n\in\{0,1\}$. By Theorem \ref{定理:恰一根}, we have $\ell_{n+1}-\ell_n=1$ for sufficiently large $n$. Hence, $\{b_n\}$ is kind of ``periodic", while $\{a_n\}$ is not. But $a_n$ can be computed from $b_n$ by solving the equation $p(a,e^a,b_n)=0$.
\end{remark}
\begin{definition}\label{定义：半周期根}
    The real solutions to equations (\ref{equ:ri}) promised by Theorem \ref{thm:semiperbunch} are called the \emph{semi-periodic solutions}. The semi-periodic solutions with respect to the triple $(\phi_i,\theta_j,\vartheta_k)$ is called a \emph{bunch} of semi-periodic solutions, or the $(\phi_i,\theta_j,\vartheta_k)$-\emph{bunch} of semi-periodic solutions.
\end{definition}

The following theorem describes the location of the semi-periodic solutions to the BETP equations (\ref{equ:ri}).
\begin{theorem}\label{定理:右上分布}
    Suppose $r,s,p,q$ and $t$ are polynomials as defined in Proposition \ref{prop:rootdecom}. Assume that the sets of root-functions $\{\phi_i\}$, $\{\theta_j\}$ and $\{\vartheta_k\}$ of the equations $p(a,y,e^a)=0$, $q(a,e^a,y)=0$ and $t(a,e^a,y)=0$, respectively, are well-defined. Then there are $N>0$ and $M>0$ such that the set of the solutions to equations (\ref{equ:ri}) in the quarter-plane
    \[
    \mathscr{Q}_{N,M}=\{(a,b)\in\mathbb{R}^2\,|\,a>N,b>M\}
    \]
    equals the union of finitely many bunches of semi-periodic solutions. Moreover, there is a subset $\{\phi_{i(\iota)}\}\subset\{\phi_i\}$, consisting of finitely many analytic algebraic exponential polynomials which are defined over $(N,\infty)$ and tending to $+\infty$ as $a\rightarrow+\infty$, such that every bunches of semi-periodic solutions lies on the graph of exactly one of these analytic algebraic exponential polynomials.
\end{theorem}
\begin{proof}
    Set $[(a_0,b_0),\ldots,(a_{s+1},b_{s+1})]$ to be an  exponential potential periodic interval set of the equation $p(a,y,e^a)=0$. Proposition \ref{thm:location-of-the-roots} indicates that there is a number $N^{(1)}$ greater than the start-point of the root-functions $\{\phi_i\}$ such that: over the interval $(N^{(1)},+\infty)$, (i) every root-function $\phi_i(a)$ bounded from above accepts the number $a_{s+1}$ as an upper bound, while (ii) every root-function $\phi_i(a)$ tending to $+\infty$ accepts the number $a_{s+1}$ as a lower bound.

    As usual, denote by $N^*$ the tripartite boundary of the equations (\ref{equ:ri}). For every triple $(\phi_i,\theta_j,\vartheta_k)$ that is of the degenerate or non-degenerate cases satisfying the condition $\lim\limits_{a\rightarrow+\infty}\phi_i(a)=+\infty$,
    Theorem \ref{thm:semiperbunch} claims that there is a number $N^s_{ijk}>N^*$ (denoted by $N^s$ in Theorem \ref{thm:semiperbunch}) such that the $(\phi_i,\theta_j,\vartheta_k)$-bunch of semi-periodic solutions, with the $a$-coordinates greater than $N^s_{ijk}$, are well-defined. Set
    \[N^{(2)}=\max_{ijk}\,N^s_{ijk},\] where the indices $ijk$ verify under the conditions that the triple $(\phi_i,\theta_j,\vartheta_k)$ is of the degenerate or non-degenerate cases and $\lim\limits_{a\rightarrow+\infty}\phi_i(a)=+\infty$. Finally, taking $M=a_{s+1}$ and $N=\max\{N^{(1)},N^{(2)}\}$, we show that the conclusions hold: 

    Denote by $\mathscr{S}_{ijk}$ the set of points $(a_*,b_*)$ with $a_*>N^*$ satisfying equations (\ref{方程:跟函数分别等于b正余弦}) defined by a triple $(\phi_i,\theta_j,\vartheta_k)$ that is of the degenerate or non-degenerate cases. Then Corollary \ref{引理:根按ijk分解} claims that 
    \[
    \mathscr{S}(N^*)=\bigcup_{ijk}\mathscr{S}_{ijk},
    \]
  where $\mathscr{S}(N^*)$ is the set of solutions $(a_*,b_*)$ to equations (\ref{equ:ri}) with $a_*>N^*$. Then we have 
  \[\mathscr{S}(N^*)\cap\mathscr{Q}_{N,M}=\bigcup_{ijk}(\mathscr{S}_{ijk}\cap\mathscr{Q}_{N,M}).\]
  Since $M=a_{s+1}$, the set $\mathscr{S}_{ijk}\cap\mathscr{Q}_{N,M}=\emptyset$ if $\phi_i$ is bounded from above by the number $M$ over the interval $(N,+\infty)$. When $\phi_i$ tends to $+\infty$, the set $\mathscr{S}_{ijk}\cap\mathscr{Q}_{N,M}$ is equal to the $(\phi_i,\theta_j,\vartheta_k)$-bunch of semi-periodic solutions (with the $a$-coordinates greater than $N$) by Theorem \ref{thm:semiperbunch}. The fact that set $\mathscr{S}_{ijk}\cap\mathscr{Q}_{N,M}$ of solutions lie on the graph of the function $y=\phi_i(a)$ is clear according to equations (\ref{方程:跟函数分别等于b正余弦}).
\end{proof}
\section{Algorithms and Experiments}\label{节：算法和实验}
\subsection{Representation and evaluation of a root-function}
We first introduce how to represent a root function: For a polynomial $g(x_1,x_2,x_3)$ satisfying $\deg(g,x_3)>0$ and $\discrim(g,x_3)\neq0$, with root-functions $\theta_1<\theta_2<\cdots<\theta_k$ sharing the start-point $M>0$, we represent the $i$-th root-function $\theta_i$ in our program by
\[
\theta_i=[[g(a,e^a,y),M,i]].
\]

Then we show how to compute a representation of a root-function and its value at some rational point: Given $g$, using \cite{strzebonski2008real} Algorithm 4.2, one computes a common upper bound $M$ for the zeros of the exponential polynomials $\lc(g,x_3)(a,a^e)$ and $\discrim(g,x_3)(a,e^a)$. Then, the number $M$ is a start-point of the root-functions of $g$ (See Theorem \ref{theorem:rootfunc}). Taking any rational number $a_0>M$ and using \cite{strzebonski2008real} Algorithm 4.5 to isolate all the real solutions to the equation 
\[
g(a_0,e^{a_0},y)=0,
\]
we then conclude that the index $i$ ranges in the set $\{1,2,\ldots,k\}$ with $k$ being the number of the isolating intervals of the above equation. Moreover, algorithms in \cite{strzebonski2008real} also allow us to refine the isolating intervals as small as we possible. This means the valuations $\theta_i(a_0),\;1\leq i\leq k,$ can be computed to arbitrary precision. In the following we denote by $[\theta_i(a_0)]$ some isolating interval of the real number $\theta_i(a_0)$, and intervals like that can be substituted into the variables of a polynomial for computation: the addition, the subtraction and the multiplication in the polynomial are understood as those between intervals, and the evaluation of that polynomial is also an interval.  
\subsection{Deciding which interval a root-function terminally belongs to}\label{子节：区间归属}
Given $I=[(a_0,b_0),\ldots,(a_{s+1},b_{s+1})]$ an 
exponential potential periodic interval set of the equation $g(a,e^a,y)=0$, Proposition \ref{thm:location-of-the-roots} claims that there is a bound $B>0$ (denoted by $M$ therein), so that for every $a>B$ and every root-function $\theta_i$ there is a unique interval $(a_{j_i},b_{j_i})$ in $I$ such that $\theta_i(a)\in(a_{j_i},b_{j_i})$. Roughly speaking, this means each root-function terminally belongs to a unique interval in $I$. Moreover, a root-function $\theta_i$ tends to $+\infty$ iff $j_i=s+1$ by Corollary \ref{coro:limits}.

The bound $B$ is an upper bound of the zeros of some exponential polynomial by Proposition \ref{thm:location-of-the-roots}. This can be computed again by \cite{strzebonski2008real} Algorithm 4.2. Then one may choose any rational number $b_0>B$ and isolate the roots of the equation 
\[g(b_0,e^{b_0},y)=0\] by intervals $T_1<T_2<\cdots<T_k$, which are small enough such that every interval $T_i$ intersects exactly one interval $(a_j,b_j)$ in $I$. Then the root-function $\theta_i$ terminally belongs to the interval $(a_j,b_j)$. In this way we can decide which root-function tends to $+\infty$ according to Corollary \ref{coro:limits}.
\subsection{The main algorithms}
In this subsection we introduce our main algorithms to count the triples of different types and to count the bunches of semi-periodic solutions.
\subsubsection{Counting triples of different types}
Algorithm \ref{算法:类型计数} counts the numbers of the  triples $(\phi_i,\theta_j,\vartheta_k)$ of different cases according to Proposition \ref{prop:rootdecom} and Definition \ref{定义:根线组合情形分类}: Steps 2-3 check whether the polynomials $p, q$ and $t$ obtained in equations (\ref{方程:pqt}) have well-defined root-functions. If not, the algorithm fails. In Step 5 we choose a rational number exceeding $M$, the common start-point of all root-functions. 

For each triple $(\phi_i,\theta_j,\vartheta_k)$ of root-functions, we compute in Step 10 an interval [det] containing the number 
\[
\det(\mathcal{M})(a_0,\phi_i(a_0),e^{a_0},\vartheta_k(a_0),\theta_j(a_0)),
\]
where the polynomial matrix $\mathcal{M}$ is as in equation (\ref{equ:idealrela}). If $0\notin$[det], then the number above is nonzero, proving that the triple $(\phi_i,\theta_j,\vartheta_k)$ is of the non-degenerate case. The counter ND, for the triples of root-functions that are of the non-degenerate case, is then increased by one (Steps 12-14).

The conditions in Step 16 imply conditions (\ref{equ:rsoFeiling}), proving that the triple $(\phi_i,\theta_j,\vartheta_k)$ is of the no-solution case. Thus the counter NS, for the triples of the no-solution case  is then increased by one (Step 17-18).

If neither of the above two circumstances is satisfied, we refine the intervals we used while deciding those conditions and see whether the results change (Steps 20-21). This will be done over and over again whenever the length of the interval [det] is greater than a preset positive number $\varepsilon$ (Step 11). We break this loop if finally we have [det]$\leq\varepsilon$. We then believe that the number 
\[
\det(\mathcal{M})(a_0,\phi_i(a_0),e^{a_0},\vartheta_k(a_0),\theta_j(a_0))
\]
may indeed equal to zero and the corresponding triple is of the degenerate case. In this case the counter D, for the triples that are of the degenerate case, is increased by one (Step 23), although the type of this triple is not decided. Therefore, when this algorithm returns the numbers ND, NS and D, we know that there are at least ND triples that are of the non-degenerate case, at least NS triples that are of the no-solution case and at most D triples that are of the degenerate case.

\begin{algorithm}\label{算法：类型计数}
\caption{CountTripleType}
\label{算法:类型计数}
\begin{algorithmic}[1]
\REQUIRE $r,s\in\Q[x_1,x_2,x_3,x_4,x_5]$, $\varepsilon>0.$
\ENSURE A vector $(\text{ND},\text{NS},\text{D})\in\N^3$ such that: \\
(1) there are at least ND triples satisfying the non-degenerate case;\\
(2) there are at least NS triples satisfying the no-solution case;\\
(3) there are at most D triples satisfying the degenerate case;
\STATE Compute polynomials $p$, $q$, $t$ according to equations (\ref{方程:pqt}).
\STATE \textbf{if} ($\deg(p,x_2)\leq0\vee\deg(q,x_5)\leq0\vee\deg(t,x_4)\leq0$) \textbf{then} \Return \text{\textcolor{red}{F}}; \textbf{end if}
\STATE \textbf{if} $0=\discrim(p,x_2)\discrim(q,x_5)\discrim(t,x_4)$ \textbf{then} \Return \text{\textcolor{red}{F}}; \textbf{end if}
\STATE Represent the root-functions $\{\phi_i\}_{i=1}^\ell$, $\{\theta_j\}_{j=1}^m$ and $\{\vartheta_k\}_{k=1}^n$ of the equations \[p(a,y,e^a)=0, q(a,e^a,y)=0\text{ and } t(a,e^a,y)=0,\] with a common start-point $M$;
\STATE Chose a rational number $a_0>M$;
\STATE Compute the polynomial matrix $\mathcal{M}$ in equation (\ref{equ:idealrela});
\STATE ND$\gets0$; NS$\gets0$; D$\gets0$;
\FOR{$(i,j,k)\in\{1,2,\ldots,\ell\}\times\{1,2,\ldots,m\}\times\{1,2,\ldots,n\}$}
\STATE Compute isolating intervals $[\phi_i(a_0)]$, $[\theta_j(a_0)]$ and $[\vartheta_k(a_0)]$;
\STATE $[\text{det}]\gets\det(\mathcal{M})(a_0,[\phi_i(a_0)],[e^{a_0}],[\vartheta_k(a_0)],[\theta_j(a_0)])$;
\WHILE{\text{Length }$[\text{det}]>\varepsilon$}
\IF{$0\notin[\text{det}]$}
\STATE ND$\gets$ ND$+1$;
\STATE \textbf{Break};
\ENDIF
\IF{$\big(\,0\notin r(a_0,[\phi_i(a_0)],[e^{a_0}],[\vartheta_k(a_0)],[\theta_j(a_0)])\;\bigvee$\\$\;\;\;\;\;\,\, 0\notin s(a_0,[\phi_i(a_0)],[e^{a_0}],[\vartheta_k(a_0)],[\theta_j(a_0)])\;\bigvee$\\$\;\;\;\;\;\,\, 1\notin [\vartheta_k(a_0)]^2+[\theta_j(a_0)]^2\,\big)$}
\STATE  NS$\gets$ NS$+1$;
\STATE \textbf{Break};
\ENDIF
\STATE Refine the intervals $[\phi_i(a_0)]$, $[e^{a_0}]$, $[\theta_j(a_0)]$ and $[\vartheta_k(a_0)]$ so that the length of each interval is reduced by at least $50\%$; 
\STATE $[\text{det}]\gets\det(\mathcal{M})(a_0,[\phi_i(a_0)],[e^{a_0}],[\vartheta_k(a_0)],[\theta_j(a_0)])$;
\ENDWHILE
\STATE \textbf{if} $\text{Length }[\text{det}]\leq\varepsilon$\textbf{ then} D$\gets \text{D}+1$;\textbf{ end if} 
\ENDFOR
\end{algorithmic}
\end{algorithm}
\subsubsection{Counting bunches of semi-periodic solutions}
A minor modification enables Algorithm \ref{算法:类型计数} to count bunches of semi-periodic solutions: In Step 4, replace the complete set $\{\phi_i\}_{i=1}^\ell$ of root-functions of the equation $p(a,y,e^a)=0$ by its subset which consists of all those root-functions tending to $+\infty$. This can be done by the method described in Subsection \ref{子节：区间归属}.

\begin{algorithm}\label{算法：半周期根计数}
\caption{CountBunch}
\label{算法：半周期根计数}
\begin{algorithmic}[1]
\REQUIRE $r,s\in\Q[x_1,x_2,x_3,x_4,x_5]$, $\varepsilon>0.$
\ENSURE A vector $(\text{ND},\text{NS},\text{D})\in\N^3$ such that: \\
(1) there are at least ND bunches of semi-periodic solutions \\$\;\;\;\;\;$given by triples of the non-degenerate case;\\
(2) there are at least NS triples of the no-solution case \\$\;\;\;\;\;$giving no solutions for sufficiently large positive $a$-coordinate;\\
(3) there are at most D bunches of semi-periodic solutions \\$\;\;\;\;\;$given by triples of the degenerate case;
\STATE $\quad\quad\quad\cdots\cdots\cdots\cdots\cdots\cdots\cdots\cdots\cdots$
    \STATE $\quad\quad\quad\cdots\text{(as in Algorithm \ref{算法：类型计数})}\cdots\cdots$
\STATE $\quad\quad\quad\cdots\cdots\cdots\cdots\cdots\cdots\cdots\cdots\cdots$
\STATE Represent all the root-functions $\{\theta_j\}_{j=1}^m$ and $\{\vartheta_k\}_{k=1}^n$ of the equations $q(a,e^a,y)=0\text{ and } t(a,e^a,y)=0,$ respectively, and represent those root-functions $\{\phi_i\}_{i=i_0}^\ell$ of the equation $p(a,y,e^a)=0$ that tends to $+\infty$,\\ such that all these root-functions share a common start-point $M$; 
\STATE $\quad\quad\quad\cdots\cdots\cdots\cdots\cdots\cdots\cdots\cdots\cdots$\\
     $\quad\quad\quad\cdots\text{(as in Algorithm \ref{算法：类型计数})}\cdots\cdots$\\ $\quad\quad\quad\cdots\cdots\cdots\cdots\cdots\cdots\cdots\cdots\cdots$
\end{algorithmic}
\end{algorithm}

If the modified algorithm (Algorithm \ref{算法：半周期根计数}) returns three numbers ND, NS and D, then by Theorem \ref{thm:semiperbunch}, we claim that (i) there are at least ND bunches of semi-periodic solutions to equations (\ref{equ:ri}) and these are given by triples of the non-degenerate case; (ii) there are at most D bunches of semi-periodic solutions to equations (\ref{equ:ri}) that are given by triples of the degenerate case; (iii) there are at least NS triples that does not give any solutions to equations (\ref{equ:ri}) when the $a$-coordinate is sufficiently large.

The experiments in the next section are done according to Algorithm \ref{算法：半周期根计数}, since we are interested mostly in the semi-periodic solutions. As we shall see, for many examples the number D$\,=0$. In this case we know there are exactly ND bunches of semi-periodic solutions.

\subsection{The experiments}
We have implemented Algorithm \ref{算法：半周期根计数} by Mathematica. 
The experiments were conducted on a laptop equipped with an AMD Ryzen 9 8945HX processor with Radeon Graphics and 32GB of onboard RAM. 

\subsubsection{Random examples}
We test Algorithm \ref{算法：半周期根计数} on some random examples. The results are shown in Table \ref{表格}. For each example in Table \ref{表格}, we randomly generate two polynomials $r$ and $s$ in the ring $\Q[x_1,x_2,x_3,x_4,x_5]$ and compute the number of bunches of semi-periodic solutions to equations (\ref{equ:ri}) by Algorithm \ref{算法：半周期根计数}. The coefficients of each randomly generated polynomial are random integers between $-10$ and $10$, the number of terms is at most $5$ and the total degree is at most $3$.

The symbol ``$\#$term" in Table \ref{表格} means the total number of terms in the randomly generated polynomials $r$ and $s$, which is at most 10 according to the settings. The meanings of the numbers ND and D are as explained in Algorithm \ref{算法：半周期根计数}. We omit the number NS since we only care about the number of bunches of semi-periodic solutions. The symbol ``F" in Table \ref{表格} means Algorithm \ref{算法：半周期根计数} returns ``F" (in the second or in the third step), the computation for that example is time-out ($>7$ hours), or an overflow error occurs.


\begin{table}\caption{Random Examples}\label{表格}
    \setlength{\abovecaptionskip}{0.1cm}
  \begin{center}
  \begin{scriptsize}
      \begin{tabular}{|c|c|c|c|c|c|}

       \cline{1-5}
       examples & $\#$term& ND& D& runtime\\ 
       \cline{1-5}
      $1$&10&0&0&\textcolor{blue}{12.74s}\\ 

     \cline{1-5}
      $2$&10&1&0&\textcolor{blue}{378.44s}\\ 

      \cline{1-5}
      $3$&10&1&0&\textcolor{blue}{159.49s}\\ 

       \cline{1-5}
      $4$&8&1&0&\textcolor{blue}{116.61s}\\ 

       \cline{1-5}
             $5$&9&2&0&\textcolor{blue}{12.03s}\\ 

       \cline{1-5}
             $6$&10&F&F&\textcolor{blue}{-}\\ 

       \cline{1-5}
             $7$&10&1&0&\textcolor{blue}{411.37s}\\ 
    \cline{1-5}
             $8$&9&0&0&\textcolor{blue}{1437.02s}\\
    \cline{1-5}
             $9$&9&0&0&\textcolor{blue}{70.8599s}\\
     \cline{1-5}
             $10$&10&0&0&\textcolor{blue}{1.89436s}\\
    \cline{1-5}
             $11$&10&0&0&\textcolor{blue}{ 593.984s}\\
  \cline{1-5}
             $12$&7&0&0&\textcolor{blue}{ 31.0385s}\\
    \cline{1-5}
             $13$&7&0&\textcolor{red}{2}&\textcolor{blue}{ 4.43889s}\\
    \cline{1-5}
             $14$&9&0&0&\textcolor{blue}{ 69.5261s}\\
    \cline{1-5}
             $15$&10&1&0&\textcolor{blue}{ 1540.32s}\\
    \cline{1-5}
             $16$&8&3&0&\textcolor{blue}{ 299.359s}\\
    \cline{1-5}
             $17$&9&F&F&\textcolor{blue}{ -}\\
    \cline{1-5}
             $18$&8&1&0&\textcolor{blue}{326.002s}\\
    \cline{1-5}
             $19$&10&0&0&\textcolor{blue}{642.32s}\\
    \cline{1-5}
             $20$&9&1&0&\textcolor{blue}{88.7287s}\\
    \cline{1-5}
             $21$&10&0&0&\textcolor{blue}{98.5475s}\\
    \cline{1-5}
             $22$&10&0&0&\textcolor{blue}{116.373s}\\
    \cline{1-5}
             $23$&9&0&0&\textcolor{blue}{1525.72s}\\
    \cline{1-5}
             $24$&8&1&0&\textcolor{blue}{32.3553s}\\
    \cline{1-5}
             $25$&9&0&0&\textcolor{blue}{14.5311s}\\
    \cline{1-5}
             $26$&10&F&F&\textcolor{blue}{-}\\
    \cline{1-5}
             $27$&9&0&0&\textcolor{blue}{28.7314s}\\
    \cline{1-5}
             $28$&10&0&0&\textcolor{blue}{23470.9s}\\
    \cline{1-5}
             $29$&10&F&F&\textcolor{blue}{-}\\
    \cline{1-5}
             $30$&8&0&0&\textcolor{blue}{1202.85s}\\
    \cline{1-5}
             $31$&9&0&0&\textcolor{blue}{1449.47s}\\
    \cline{1-5}
             $32$&8&0&0&\textcolor{blue}{6.62692s}\\
    \cline{1-5}
             $33$&9&0&0&\textcolor{blue}{164.036s}\\
    \cline{1-5}
             $34$&9&0&0&\textcolor{blue}{762.245s}\\
    \cline{1-5}
             $35$&10&F&F&\textcolor{blue}{-}\\
    \cline{1-5}
             $36$&8&1&0&\textcolor{blue}{77.698s}\\
    \cline{1-5}
             $37$&10&0&0&\textcolor{blue}{354.641s}\\
    \cline{1-5}
             $38$&8&3&0&\textcolor{blue}{1691.75s}\\
    \cline{1-5}
             $39$&9&F&F&\textcolor{blue}{-}\\
    \cline{1-5}
             $40$&9&0&0&\textcolor{blue}{2629.38s}\\
    \cline{1-5}
             $41$&9&0&0&\textcolor{blue}{373.463s}\\
    \cline{1-5}
             $42$&10&0&0&\textcolor{blue}{26.6591s}\\
    \cline{1-5}
             $43$&9&F&F&\textcolor{blue}{-}\\
    \cline{1-5}
             $44$&8&F&F&\textcolor{blue}{-}\\
    \cline{1-5}
             $45$&10&1&0&\textcolor{blue}{155.927s}\\
    \cline{1-5}
             $46$&9&0&0&\textcolor{blue}{2.31042s}\\
    \cline{1-5}
             $47$&8&0&0&\textcolor{blue}{13.8368s}\\
    \cline{1-5}
             $48$&10&0&0&\textcolor{blue}{278.94s}\\
    \cline{1-5}
             $49$&10&F&F&\textcolor{blue}{-}\\
    \cline{1-5}
             $50$&10&0&0&\textcolor{blue}{123.977s}\\
    \cline{1-5}
             $51$&9&2&0&\textcolor{blue}{614.898s}\\
    \cline{1-5}
             $52$&9&0&0&\textcolor{blue}{237.112s}\\
    \cline{1-5}
             $53$&9&0&0&\textcolor{blue}{2.72436s}\\
    \cline{1-5}
             $54$&10&F&F&\textcolor{blue}{-}\\
    \cline{1-5}
             $55$&8&F&F&\textcolor{blue}{-}\\
    \cline{1-5}
             $56$&8&1&0&\textcolor{blue}{457.403s}\\
    \cline{1-5}
             $57$&9&0&0&\textcolor{blue}{943.769s}\\
    \cline{1-5}
             $58$&10&0&0&\textcolor{blue}{984.132s}\\
    \cline{1-5}
             $59$&10&F&F&\textcolor{blue}{-}\\
    \cline{1-5}
             $60$&9&0&0&\textcolor{blue}{2048.26s}\\
    \hline

    \end{tabular}
    \end{scriptsize}
  \end{center}
 \end{table}
     The results suggest that Algorithm \ref{算法：半周期根计数} can effectively determine the exact number of bunches of semi-periodic solutions for the vast majority of randomly generated examples: Only one example involves undetermined triples (i.e., with D$\,>0$) and there are only 12 examples out of 60 returning ``F". However, the running times differ significantly although these examples have similar size: The shortest and the longest running time is, respectively, $1.89436\,$s and $23470.9\,$s. This indicates that the running time of the algorithm depends not only on the size of the example.
     
     \subsubsection{Examples coming from MTP's}
     The problem of computing complex roots of an MTP reduces to computing real solutions of equations (\ref{equ:ri}), which are derived from that MTP. Table \ref{MTP例子} shows some examples of this type.

     The symbol ttlDeg in Table \ref{MTP例子} stands for the greater of the total degrees of the two polynomials derived from the MTP, while the other symbols have the same meaning as in Table \ref{表格}.

     These examples have higher total degrees and more terms than those randomly generated examples in the last sub-subsection. As usual, D$\,=0$ for most examples, meaning that we get the exact number of bunches of semi-periodic solutions for these examples. On the other hand, the running time scales noticeably with the size of the example: it is within 30 seconds for smaller examples (1, 2, 3 and 6) and more than 1000 seconds for larger ones (4, 5, 7 and 8). To conclude, Algorithm \ref{算法：半周期根计数} is effective for small examples derived from MTP's.
\begin{table}\caption{Examples coming from MTP's}\label{MTP例子}
    \setlength{\abovecaptionskip}{0.1cm}
  \begin{center}
  \begin{scriptsize}
      \begin{tabular}{|c|c|c|c|c|c|c|c|c|c|c|}

       \hline
       &MTP's & ttlDeg&$\#$term& ND & D & runtime\\
       \hline
          1& $2 z - 5 \sin z + 7 \cos z$&4&12&1&0&\textcolor{blue}{13.61s}\\
    \hline
    2&$8z + 12 \sin z - \cos z + 7$&4&14&1&0&\textcolor{blue}{27.00s}\\
    \hline
       3& $-8 z + 17 \sin z + 45 \cos z + 4$&4&14&1&0&\textcolor{blue}{15.28s}\\
    \hline
4&$5 z\sin z + \cos^2 z + 1$&8&21&1&0&\textcolor{blue}{1164.23s}\\
    \hline
    5&$3 z - \sin^2 z + 2\cos z $&8&21&2&0&\textcolor{blue}{4120.05s}\\
    \hline
   6& $2 z\cos z - 5 \sin z + 7 \cos z $&5&16&0&\textcolor{red}{2}&\textcolor{blue}{7.25s}\\
    \hline
      7&  $2 z\cos z - 5 \sin z + 7 \sin z\cos z $&8&24&1&0&\textcolor{blue}{1335.78s}\\
    \hline
         8&   $2 z\cos z - 5 \sin z + 7 \cos^2 z $&8&27&1&0&\textcolor{blue}{1075.45s}\\
    \hline
    \end{tabular}
    \end{scriptsize}
  \end{center}
\end{table}
\section{Conclusion}\label{节：总结}
This paper describes the real solutions of the BETP equations 
\[
\left\{
\begin{array}{rcl}
     r(a,b,e^a,\sin b, \cos b)&=&0  \\
     s(a,b,e^a,\sin b, \cos b)&=&0
\end{array}
\right.
\]
in the unbounded quarter 
\[
\{(a,b)\in\R^2\,|\,a>N,b>M\}
\]
as some bunches of semi-periodic solutions on some curves, when the positive numbers $N$ and $M$ are sufficiently large. Although this description is neat in theory, there are two shortcomings: 1) The number $N$ and $M$ are not given explicitly. 2) We don't know how to certify a degenerate triple of root-functions.
\bibliographystyle{elsarticle-num}
\bibliography{ref}

@inproceedings{achatz2008deciding,
  title={Deciding polynomial-exponential problems},
  author={Achatz, Melanie and McCallum, Scott and Weispfenning, Volker},
  booktitle={Proceedings of the International Symposium on Symbolic and Algebraic Computation},
  pages={215--222},
  year={2008}
}

@inproceedings{mitchell1990robust,
  title={Robust ray intersection with interval arithmetic},
  author={Mitchell, Don P},
  booktitle={Proceedings of Graphics Interface},
  volume={90},
  pages={68--74},
  year={1990}
}

@inproceedings{strzebonski2009real,
  title={Real root isolation for tame elementary functions},
  author={Strzebonski, Adam},
  booktitle={Proceedings of the  International Symposium on Symbolic and Algebraic Computation},
  pages={341--350},
  year={2009}
}

@article{strzebonski2011cylindrical,
  title={Cylindrical decomposition for systems transcendental in the first variable},
  author={Strzebo{\'n}ski, Adam},
  journal={Journal of Symbolic Computation},
  volume={46},
  number={11},
  pages={1284--1290},
  year={2011},
  publisher={Elsevier}
}

@article{wang2021symbolic,
  title={A symbolic-numerical algorithm for isolating real roots of certain radical expressions},
  author={Wang, Dongming and Xu, Juan},
  journal={Journal of Computational and Applied Mathematics},
  volume={391},
  pages={113424},
  year={2021},
  publisher={Elsevier}
}

@inproceedings{maignan1998solving,
  title={Solving one and two-dimensional exponential polynomial systems},
  author={Maignan, Aude},
  booktitle={Proceedings of the  International Symposium on Symbolic and Algebraic Computation},
  pages={215--221},
  year={1998}
}

@inproceedings{collins1976polynomial,
  title={Polynomial real root isolation by differentiation},
  author={Collins, George E and Loos, R{\"u}diger},
  booktitle={Proceedings of the 3rd ACM Symposium on Symbolic and Algebraic Computation},
  pages={15--25},
  year={1976}
}

@inproceedings{strzebonski2008real,
  title={Real root isolation for exp-log functions},
  author={Strzebo{\'n}ski, Adam},
  booktitle={Proceedings of the  International Symposium on Symbolic and Algebraic Computation},
  pages={303--314},
  year={2008}
}

@incollection{mccallum1998improved,
  title={An improved projection operation for cylindrical algebraic decomposition},
  author={McCallum, Scott},
  booktitle={Quantifier Elimination and Cylindrical Algebraic Decomposition},
  pages={242--268},
  year={1998},
  publisher={Springer}
}

@inproceedings{richardson1991towards,
  title={Towards computing non algebraic cylindrical decompositions},
  author={Richardson, Daniel},
  booktitle={Proceedings of the 1991 international symposium on Symbolic and algebraic computation},
  pages={247--255},
  year={1991}
}

@article{xu2015quantifier,
  title={Quantifier elimination for a class of exponential polynomial formulas},
  author={Xu, Ming and Li, Zhi-Bin and Yang, Lu},
  journal={Journal of Symbolic Computation},
  volume={68},
  pages={146--168},
  year={2015},
  publisher={Elsevier}
}

@article{huang2018positive,
  title={Positive root isolation for poly-powers by exclusion and differentiation},
  author={Huang, Cheng-Chao and Li, Jing-Cao and Xu, Ming and Li, Zhi-Bin},
  journal={Journal of Symbolic Computation},
  volume={85},
  pages={148--169},
  year={2018},
  publisher={Elsevier}
}

@book{shidlovskii2011transcendental,
author = {Andrei B. Shidlovskii},
title = {Transcendental Numbers},
year = {2011},
publisher = {De Gruyter},
}

@article{mccallum2012deciding,
  title={Deciding polynomial-transcendental problems},
  author={McCallum, Scott and Weispfenning, Volker},
  journal={Journal of Symbolic Computation},
  volume={47},
  number={1},
  pages={16--31},
  year={2012},
  publisher={Elsevier}
}

@article{strzebonski2012real,
  title={Real root isolation for exp--log--arctan functions},
  author={Strzebo{\'n}ski, Adam},
  journal={Journal of Symbolic Computation},
  volume={47},
  number={3},
  pages={282--314},
  year={2012},
  publisher={Elsevier}
}

@article{lipparini2024satisfiability,
  title={Satisfiability modulo Nonlinear Arithmetic and Transcendental Functions via Numerical and Topological methods},
  author={Lipparini, Enrico},
  year={2024},
  publisher={Universit{\`a} degli studi di Genova}
}

@inproceedings{ni2024local,
  title={Local Search for Checking Satisfiability of Formulas with Trigonometric Functions},
  author={Ni, Xinpeng and Xia, Bican and Zhao, Tianqi},
  booktitle={International Symposium on Automated Technology for Verification and Analysis},
  pages={256--274},
  year={2024},
  organization={Springer}
}

@article{lipparini2025satisfiability,
  title={Satisfiability of non-linear transcendental arithmetic as a certificate search problem},
  author={Lipparini, Enrico and Ratschan, Stefan},
  journal={Journal of Automated Reasoning},
  volume={69},
  number={1},
  pages={3},
  year={2025},
  publisher={Springer}
}

@inproceedings{chen2024reduction,
  title={Reduction of transcendental decision problems over the reals},
  author={Chen, Rizeng and Xia, Bican},
  booktitle={Proceedings of the 2024 International Symposium on Symbolic and Algebraic Computation},
  pages={56--64},
  year={2024}
}

@article{matsumoto2012nonlinear,
  title={Nonlinear delay monopoly with bounded rationality},
  author={Matsumoto, Akio and Szidarovszky, Ferenc},
  journal={Chaos, Solitons \& Fractals},
  volume={45},
  number={4},
  pages={507--519},
  year={2012},
  publisher={Elsevier}
}

@article{chen2024isolating,
  title={Isolating all the real roots of a mixed trigonometric-polynomial},
  author={Chen, Rizeng and Li, Haokun and Xia, Bican and Zhao, Tianqi and Zheng, Tao},
  journal={Journal of Symbolic Computation},
  volume={121},
  pages={102250},
  year={2024},
  publisher={Elsevier}
}
}
\end{document}